\documentclass{llncs}

\usepackage{amscd,latexsym}
\usepackage{amsmath}
\usepackage{amsfonts}
\usepackage{amssymb}
\usepackage{wasysym}
\usepackage{algorithm}
\usepackage{algorithmic}

\usepackage{times}
\usepackage{paralist}
\usepackage{fullpage}

\newcommand{\structure}[1]{\left \langle #1 \right \rangle}
\newcommand{\set}[1]{\left \{ #1 \right \}}

\def\qed{\rule{0.4em}{1.4ex}}

\newcommand{\slopefrac}[2]{\leavevmode\kern.1em
  \raise .5ex\hbox{\the\scriptfont0 #1}\kern-.1em
  /\kern-.15em\lower .25ex\hbox{\the\scriptfont0 #2}}

\makeatletter
\begingroup \catcode `|=0 \catcode `[= 1
\catcode`]=2 \catcode `\{=12 \catcode `\}=12
\catcode`\\=12 |gdef|@xcomment#1\end{comment}[|end[comment]]
|endgroup

\def\@comment{\let\do\@makeother \dospecials\catcode`\^^M=10\def\par{}}

\def\begincomment{\@comment\@xcomment}

\makeatother

\newenvironment{comment}{\begincomment}{}

\begin{document}
\title{Faster Algorithms for Alternating Refinement Relations}
\author{Krishnendu Chatterjee\inst{1} \and Siddhesh Chaubal\inst{2} \and Pritish Kamath\inst{2}}
\institute{IST Austria (Institute of Science and Technology, Austria) \and IIT Bombay}

\maketitle

\newif
  \iflong
  \longfalse
\newif
  \ifshort
  \shorttrue

\thispagestyle{empty}

\begin{abstract}
One central issue in the formal design and analysis of reactive systems is 
the notion of \emph{refinement} that asks whether all behaviors of the 
implementation is allowed by the specification.
The local interpretation of behavior leads to the notion of \emph{simulation}.
Alternating transition systems (ATSs) provide a general model for composite 
reactive systems, and the simulation relation for  ATSs is known as alternating 
simulation.
The simulation relation for fair transition systems is called fair simulation.
In this work our main contributions are as follows: 
(1)~We present an improved algorithm for fair simulation with B\"uchi fairness
constraints; our algorithm requires  $O(n^3 \cdot m)$ time as compared to the 
previous known $O(n^6)$-time algorithm, where $n$ is the number of states and 
$m$ is the number of transitions.
(2)~We present a game based algorithm for alternating simulation that 
requires $O(m^2)$-time as compared to the previous known 
$O((n \cdot m)^2)$-time algorithm, where $n$ is the number of states and 
$m$ is the size of transition relation.
(3)~We present an iterative algorithm for alternating simulation that matches
the time complexity of the game based algorithm, but is more space efficient
than the game based algorithm.
\end{abstract}

\section{Introduction}

\noindent{\bf Simulation relation and extensions.}
One central issue in formal design and analysis of reactive systems is the 
notion of refinement relations.
The refinement relation (system $A$ refines system $A'$) intuitively means 
that every behavioral option of $A$ (the implementation) is allowed by 
$A'$ (the specification).
The \emph{local} interpretation of behavorial option in terms of successor 
states leads to refinement as \emph{simulation}~\cite{Milner71}.
The simulation relation enjoys many appealing properties, such as it has a 
denotational characterization, it has a logical characterization and 
it can be computed in polynomial time (as compared to trace containment 
which is PSPACE-complete).
While the notion of simulation was originally developed for transition 
systems~\cite{Milner71}, it has many important extensions.
Two prominent extensions are as follows: 
(a)~extension for composite systems and (b)~extension for fair transition 
systems. 

\smallskip\noindent{\bf Alternating simulation relation.}
Composite reactive systems can be viewed as multi-agent 
systems~\cite{Sha53,HF89},
where each possible step of the system corresponds to a possible move in 
a game which may involve some or all component moves. 
We model multi-agent systems as \emph{alternating transition systems} (ATSs)~\cite{AHK02}.
In general a multi-agent system consists of a set $I$ of agents, but for 
algorithmic purposes for simulation we always consider a subset 
$I' \subseteq I$ of agents against the rest, and thus we will only consider
two-agent systems (one agent is the collection $I'$ of agents, and the other 
is the collection of the rest of the agents). 
Consider the composite systems $A || B$ and $A' || B$, in environment $B$.
The relation that $A$ refines $A'$ without constraining the environment $B$ 
is expressed by generalizing the simulation relation to \emph{alternating 
simulation relation}~\cite{AHKV98}.
Alternating simulation also enjoys the appealing properties of 
denotational and logical characterization along with polynomial time 
computability.
We refer the readers to~\cite{AHKV98} for an excellent exposition of 
alternating simulation and its applications in design and analysis of 
composite reactive systems. 
We briefly discuss some applications of alternating simulation relation. 
Given a composite system with many components, the problem of refinement
of a component (i.e., a component $C$ can be replaced with its implementation
$C'$) without affecting the correctness of the composite system is an 
alternating simulation problem.
Similarly, refinement for open reactive systems is also an alternating 
simulation problem.
Finally, graph games provide the mathematical framework to analyze many
important problems in computer science, specially in relation to logic,
as there is a close connection of automata and graph games 
(see~\cite{Thomas97,GH82} for details).
Alternating simulation provides the technique for state space reduction for 
graph games, which is a pre-requisite for efficient algorithmic analysis of 
graph games.
Thus computing alternating simulation for ATSs is a core algorithmic question 
in the formal analysis 
of composite systems, as well as in the heart of efficient algorithmic 
analysis of problems related to logic in computer science.

\smallskip\noindent{\bf Fair simulation relation.}
Fair transition systems are extension of transition systems with fairness 
constraint. 
A \emph{liveness} (or weak fairness or B\"uchi fairness) constraint 
consists of a set $B$ of live states, and requires that runs of the system
visit some live state infinitely often. 
In general the fairness constraint can be a strong fairness constraint instead
of a liveness constraint.
The notion of simulation was extended to fair transition systems as 
\emph{fair simulation}~\cite{HKR97}.
It was shown in~\cite{HKR97} that fair simulation also enjoys the appealing 
properties of denotational and logical characterization, and polynomial time 
computability (see~\cite{HKR97} for many other important properties and 
discussion on fair simulation).
Again the computation of fair simulation with B\"uchi fairness constraints
is an important algorithmic problem for design and analysis of reactive systems 
with liveness requirements.

\smallskip\noindent{\bf Our contributions.} 
In this work we improve the algorithmic complexities of computing 
fair simulation with B\"uchi fairness constraints and 
alternating simulation.
In the descriptions below we will denote by $n$ the size of the state space
of systems, and by $m$ the size of the transition relation.
Our main contributions are summarized below.
\begin{compactenum}
\item \emph{Fair simulation.} 
First we extend the notion of fair simulation to alternating fair simulation
for ATSs with B\"uchi fairness constraints.
There are two natural ways of extending the definition of fair simulation 
to alternating fair simulation, and we show that both the definitions 
coincide. 
We present an algorithm to compute the alternating fair simulation relation
by a reduction to a game with parity objectives with three priorities. 
As a special case of our algorithm for fair simulation, we show that the 
fair simulation relation can be computed in $O(n^3 \cdot m)$ time, as 
compared to the previous known $O(n^6)$-time algorithm of~\cite{HKR97}.
Observe that $m$ is at most $O(n^2)$ and thus the worst case running time
of our algorithm is $O(n^5)$. 
Moreover, in many practical examples systems have constant out-degree
(for examples see~\cite{ClarkeBook})  
(i.e., $m=O(n)$), and then our algorithm requires $O(n^4)$ time.

\item \emph{Game based alternating simulation.} 
We present a game based algorithm for alternating simulation. 
Our algorithm is based on a reduction to a game with reachability objectives,
and requires $O(m^2)$ time, as compared to the previous known 
algorithm that requires $O((n \cdot m)^2)$ time~\cite{AHKV98}. 
One key step of the reduction is to construct the game graph in time 
linear in the size of the game graph. 

\item \emph{Iterative algorithm for alternating simulation.} 
We present an iterative algorithm to compute the alternating simulation 
relation.
The time complexity of the iterative algorithm matches the 
time complexity of the game based algorithm, however, the iterative algorithm 
is more space efficient.
(see paragraph on space complexity of Section~\ref{subsec:iterative} for the 
detailed comparision).
Moreover, both the game based algorithm and the iterative algorithm when 
specialized to transition systems match the best known algorithms to compute 
the simulation relation.

\end{compactenum}

We remark that the game based algorithms we obtain for alternating fair 
simulation and alternating simulation are reductions to standard two-player 
games on graphs with parity objectives (with three priorities) and reachability
objectives. 
Since such games are well-studied, standard algorithms developed for games 
can now be used for computation of refinement relations. 
Our key technical contribution is establishing the correctness of the 
efficient reductions, and showing that the game graphs can be constructed in 
linear time in the size of the game graphs. 
For the iterative algorithm we establish an alternative characterization of
alternating simulation, and present an iterative algorithm that simultaneously
prunes two relations, without explicitly constructing game graphs (thus 
saving space), to compute the relation obtained by the alternative 
characterization.

\newcommand{\wh}{\widehat}
\newcommand{\ov}{\overline}
\newcommand{\wt}{\widetilde}
\newcommand{\whw}{\wh{w}}
\newcommand{\Inf}{\mathsf{Inf}}
\newcommand{\fair}{\mathsf{fair}}
\newcommand{\fairalt}{\mathsf{fairalt}}
\newcommand{\weak}{\mathsf{weak}}
\newcommand{\strong}{\mathsf{strong}}
\newcommand{\altsim}{\mathsf{altsim}}
\newcommand{\sseq}{\langle v_0,v_1,v_2,\ldots\rangle}

\newcommand{\pat}{\omega}
\newcommand{\Pat}{\Omega}
\newcommand{\straa}{\alpha}
\newcommand{\Straa}{\mathcal{A}}
\newcommand{\strab}{\beta}
\newcommand{\Strab}{\mathcal{B}}
\newcommand{\Reach}[1]{\mathrm{Reach}(#1)}
\newcommand{\Safe}[1]{\mathrm{Safe}(#1)}
\newcommand{\Parity}{\mathrm{Parity}}
\newcommand{\ReachSafe}[1]{\mathrm{ReachSafe}(#1)}
\newcommand{\Buchi}[1]{\mathrm{Buchi}(#1)}
\newcommand{\coBuchi}[1]{\mathrm{coBuchi}(#1)}
\newcommand{\In}{{\mathsf In}}
\newcommand{\Out}{{\mathsf Out}}

\section{Definitions}\label{sec:defn}
In this section we present all the relevant definitions, and the previous best 
known results.
We present definitions of labeled transition systems (Kripke structures), 
labeled alternating transitions systems (ATS), fair simulation, 
and alternating simulation.
All the simulation relations we will define are closed under union 
(i.e., if two relations are simulation relations, then so is their 
union), and we will consider the maximum simulation relation.
We also present relevant definitions for graph games that will be later 
used for the improved results.

\begin{definition}[Labeled transition systems (TS)]
A labeled \emph{transition system (TS)} (Kripke structure) is a tuple  
$K=\structure{\Sigma, W, \whw, R, L}$, where $\Sigma$ is a finite set
of observations; $W$ is a  finite set of states and $\whw$ is the 
initial state; $R \subseteq W \times W$ is the transition relation; 
and $L: W \rightarrow \Sigma$ is the labeling function that maps each state 
to an observation.
For technical convenience we assume that for all $w\in W$ there exists
$w'\in W$ such that $(w,w') \in R$.
\end{definition}

\smallskip\noindent{\em  Runs, fairness constraint, and fair transition 
systems.}
For a TS $K$ and a state $w \in W$, a $w$-run of $K$ is an infinite sequence 
$\ov{w}=w_0, w_1, w_2, \ldots $ of states  such that 
$w_0 = w $ and $R(w_i,w_{i+1})$ for all $i \ge 0$. 
We write $\Inf(\ov{w})$ for the set of states that occur infinitely often in 
the run $\ov{w}$. 
A run of $K$ is a $\whw$-run for the initial state $\whw$.
In this work we will consider \emph{B\"uchi fairness constraints}, and 
a B\"uchi fairness constraint is specified as a set $F \subseteq W$ 
of B\"uchi states, and defines the fair set of runs, where a run $\ov{w}$ is 
\emph{fair} iff $\Inf(\ov{w}) \cap F \neq \emptyset$ (i.e., the run 
visits $F$ infinitely often).
A \emph{fair transition system} 
$\mathcal{K} = \structure{K, F}$ consists of a TS $K$ and a B\"uchi fairness constraint $F\subseteq W$ for $K$. 
We consider two TSs $K_1 = \structure{\Sigma, W_1, \whw_1, R_1, L_1}$ and 
$K_2 = \structure{\Sigma, W_2, \whw_2, R_2, L_2}$ over the same alphabet, 
and the two fair TSs $\mathcal{K}_1 = \structure{K_1, F_1}$ and 
$\mathcal{K}_2 = \structure{K_2,F_2}$.
We now define the fair simulation between $\mathcal{K}_1$ and 
$\mathcal{K}_2$~\cite{HKR97}.

\begin{definition}[Fair simulation]\label{defn:fair_sim}
A binary relation $S \subseteq W_1 \times W_2$ is a \emph{fair simulation} of $\mathcal{K}_1$ by $\mathcal{K}_2$ if the following two conditions hold
for all $(w_1,w_2) \in W_1 \times W_2$:
\begin{enumerate} 
\item If $S(w_1,w_2)$, then $L_1(w_1) = L_2(w_2)$.
\item There exists a strategy $\tau : (W_1 \times W_2)^+ \times W_1 \rightarrow W_2$ such that if $S(w_1,w_2)$ and $\ov{w} = u_0, u_1, u_2,\ldots$ 
is a fair $w_1$-run of $\mathcal{K}_1$, then the following conditions hold:
(a)~the outcome $\tau[\ov{w}] = u_0',u_1',u_2', \ldots$ is a fair $w_2$-run of 
$\mathcal{K}_2$ (where the outcome $\tau[\ov{w}]$ is defined as follows: 
for all $i\geq 0$ we have 
$u_i'=\tau((u_0,u_0'),(u_1,u_1'),\ldots, (u_{i-1},u_{i-1}'), u_i)$) ; and 
(b)~the outcome $\tau[\ov{w}]$  $S$-matches $\ov{w}$; that is, $S(u_i,u_i')$ 
for all $i \geq 0$.
We say $\tau$ is a \emph{witness} to the fair simulation $S$.
\end{enumerate}
We denote by $\preceq_{\fair}$ the maximum fair simulation relation between 
$\mathcal{K}_1$ and $\mathcal{K}_2$.
We say that the fair TS $\mathcal{K}_2$ \emph{fairly simulates} the 
fair TS $\mathcal{K}_1$ iff $(\whw_1,\whw_2) \in \preceq_{\fair}$.
\end{definition}

We have the following result for fair simulation from~\cite{HKR97} (see item~1 of
Theorem 4.2 from~\cite{HKR97}).

\begin{theorem}\label{thrm:fair}
Given two fair TSs $\mathcal{K}_1$ and $\mathcal{K}_2$, the problem of 
whether $\mathcal{K}_2$ fairly simulates $\mathcal{K}_1$ 
can be decided in time 
$O((|W_1| + |W_2|)\cdot (|R_1| + |R_2|) + (|W_1|\cdot |W_2|)^3)$.
\end{theorem}

\begin{definition}[Labeled alternating transition systems (ATS)]
A labeled \emph{alternating transitions system (ATS)} is a tuple 
$K = \structure{\Sigma,W,\whw,A_1,A_2,P_1,P_2,L,\delta}$, where
(i)~$\Sigma$ is a finite set of observations; (ii)~$W$ is a finite set of states 
with $\whw$ the initial state;
(iii)~$A_i$ is a finite set of actions for Agent $i$, for $i \in \set{1,2}$; 
(iv)~$P_i : W \rightarrow 2^{A_i} \setminus \emptyset$ assigns to every state 
$w$ in $W$ the non-empty set of actions available to Agent~$i$ at $w$, for 
$i \in \set{1,2}$; (v)~$L : W \rightarrow \Sigma$ is the labeling function that 
maps every state to an observation; and 
(vi)~$\delta : W \times A_1 \times A_2 \rightarrow W$ is the transition 
relation that given a state and the joint actions gives the next state.
\end{definition}

Observe that a TS can be considered as a special case of ATS with $A_2$ 
singleton (say $A_2=\set{\bot}$), and the transition relation $R$ of a TS is 
described by the transition relation $\delta:W \times A_1 \times \set{\bot}
\rightarrow W$ of the ATS.

\begin{definition}[Alternating simulation]
Given two ATS, $K = \structure{\Sigma,W,\whw,A_1,A_2,P_1,P_2,L,\delta}$
and $K' = \structure{\Sigma,W',\whw',A_1',A_2',P_1',P_2',L',\delta'}$
a binary relation $S \subseteq W \times W'$ is an alternating simulation from 
$\mathcal{K}$ to $\mathcal{K}'$ 
if for all states $w$ and $w'$ with $(w,w') \in S$, the following conditions hold :
\begin{enumerate}
\item $L(w) = L'(w')$
\item For every action $a \in P_1(w)$, there exists an action $a' \in P_1'(w')$ such that for every action $b' \in P_2'(w')$, 
there exists an action $b \in P_2(w)$ such that 
$(\delta(w,a,b),\delta'(w',a',b')) \in S$, i.e.,
      \begin{equation*}
       \forall (w,w') \in S \cdot \  \forall a \in P_1(w) \cdot \exists a' \in P_1'(w') \cdot \forall b' \in P_2'(w') \cdot \exists b \in P_2(w) \cdot (\delta(w,a,b),\delta'(w',a',b')) \in S
      \end{equation*}
\end{enumerate}
We denote by $\preceq_{\altsim}$ the maximum alternating simulation relation 
between $K$ and $K'$. We say that the ATS $K'$ \emph{simulates} the 
ATS $K$ iff $(\whw,\whw') \in \preceq_{\altsim}$.
\end{definition}

The following result was shown in~\cite{AHKV98} (see proof of Theorem~3 of~\cite{AHKV98}).

\begin{theorem}\label{thrm:altsim}
For two ATSs $K$ and $K'$, the alternating simulation relation $\preceq_{\altsim}$ 
can be computed in time 
$O(|W|^2\cdot |W'|^2\cdot |A_1|\cdot |A_1'|\cdot |A_2|\cdot |A_2'|)$.
\end{theorem}

In the following section we will present an extension of the notion of 
fair simulation for TSs to alternating fair simulation for ATSs, and present 
improved algorithms to compute $\preceq_{\fair}$ and $\preceq_{\altsim}$.
Some of our algorithms will be based on reduction to two-player games 
on graphs. We present the required definitions below.

\smallskip\noindent{\bf Two-player Game graphs.} 
A \emph{two-player game graph} $G=((V,E),(V_1,V_2))$ consists of a 
directed graph $(V,E)$ with a set $V$ of $n$ vertices and a set $E$ of 
$m$ edges, and a partition $(V_1,V_2)$ of $V$ into two sets.
The vertices in $V_1$ are {\em player~1 vertices},
where player~1 chooses the outgoing edges; 
and the vertices in $V_2$ are {\em player~2 vertices}, 
where  player~2 (the adversary to player~1) chooses the outgoing edges.
For a vertex $u\in V$, we write $\Out(u)=\set{v\in V \mid (u,v) \in E}$ 
for the set of successor vertices of~$u$ and $\In(u)=\set{v \in V \mid (v,u) \in E}$ 
for the set of incoming edges of $u$.
We assume that every vertex has at least one out-going edge. 
i.e., $\Out(u)$ is non-empty for all vertices $u\in V$.

\noindent{\em Plays.} A game is played by two players: 
player~1 and player~2, who form an infinite path in the game graph by 
moving a token along edges.
They start by placing the token on an initial vertex, and then they
take moves indefinitely in the following way.
If the token is on a vertex in~$V_1$, then player~1 moves the token along
one of the edges going out of the vertex.
If the token is on a vertex in~$V_2$, then player~2 does likewise.
The result is an infinite path in the game graph, called a {\em play}.
We write $\Pat$ for the set of all plays.

\noindent{\em Strategies.} 
A strategy for a player is a rule that specifies how to extend plays.
Formally, a \emph{strategy} $\straa$ for player~1 is a function 
$\straa$: $V^* \cdot V_1 \to V$ such that for all $w \in V^*$ and 
all $v \in V_1$ we have $\straa(w \cdot v) \in \Out(v)$, 
and analogously for player~2 strategies.
We write $\Straa$ and $\Strab$ for the sets of all strategies for 
player~1 and player~2, respectively.
A \emph{memoryless} strategy for player~1 is independent of the 
history and depends only on the current state, and can be described 
as a function $\straa: V_1 \to V$, and similarly for player~2.
Given a starting vertex $v\in V$, a strategy $\straa\in\Straa$ for player~1, 
and a strategy $\strab\in\Strab$ for player~2, there is a unique play, 
denoted $\pat(v,\straa,\strab)=\sseq$, which is defined as follows: 
$v_0=v$ and for all $k \geq 0$,
if $v_k \in V_1$, then $\straa(v_k)=v_{k+1}$, and
if $v_k \in V_2$, then $\strab(v_k)=v_{k+1}$.
We say a play $\pat$ is \emph{consistent} with a strategy of a player, 
if there is a strategy of the opponent such that given both the 
strategies the unique play is $\pat$. 

\noindent{\em Objectives.} An objective $\Phi$ for a game graph is a 
desired subset of plays.
For a play $\pat = \sseq\in \Omega$,  we define $\Inf(\pat) = 
\set{v \in V \mid \mbox{$v_k = v$ for infinitely many $k \geq 0$}}$
to be the set of vertices that occur infinitely often in~$\pat$.
We define reachability, safety and parity objectives with three 
priorities.
\begin{enumerate}

 \item 
  \emph{Reachability and safety objectives.}
  Given a set $T \subseteq V$ of vertices, the reachability objective 
  $\Reach{T}$ requires that some vertex in $T$ be visited,
  and dually, 
  the safety objective $\Safe{F}$ requires that only vertices in $F$ 
  be visited.
  Formally, the sets of winning plays are
  $\Reach{T}= \set{\sseq \in \Pat \mid 
  \exists k \geq 0. \ v_k \in T}$
  and 
  $\Safe{F}=\set{\sseq \in \Pat \mid 
  \forall k \geq 0.\ v_k \in F}$.
  The reachability and safety objectives are dual in the sense that 
  $\Reach{T}= \Pat \setminus \Safe{V \setminus T}$.

\item
  \emph{Parity objectives with three priorities.}
  Consider a priority function $p:V \to \set{0,1,2}$ that maps
  every vertex to a priority either~0, 1 or~2. 
  The parity objective requires that the minimum priority visited
  infinitely often is even. 
  In other words, the objectives require that either vertices 
  with priority~0 are visited infinitely often, or vertices with 
  priority~1 are visited finitely often.
  Formally the set of winning plays is $\Parity(p)=
  \set{ \pat \mid \Inf(\pat) \cap p^{-1}(0) \neq \emptyset \text{ or } 
   \Inf(\pat) \cap p^{-1}(1) = \emptyset}$.

\end{enumerate}

\smallskip\noindent{\em Winning strategies and sets.}
Given an objective $\Phi\subseteq\Pat$ for player~1, a strategy 
$\straa\in\Straa$ is a \emph{winning strategy}
for player~1 from a vertex $v$ if for all player~2 strategies $\strab\in\Strab$ 
the play $\pat(v,\straa,\strab)$ is winning, i.e., 
$\pat(v,\straa,\strab) \in \Phi$.
The winning strategies for player~2 are defined analogously by switching the 
role of player~1 and player~2 in the above definition.
A vertex $v\in V$ is winning for player~1 with respect to the objective 
$\Phi$ if player~1 has a winning strategy from $v$.
Formally, the set of \emph{winning vertices for player~1} with respect to 
the objective $\Phi$ is the set
$W_1(\Phi) =\set{v \in V \mid \exists \straa\in\Straa. 
\ \forall \strab\in\Strab.\ \pat(v,\straa,\strab) \in \Phi}$.
Analogously, the set of all winning vertices for player~2 with respect to an 
objective $\Psi\subseteq\Pat$ is 
$W_2(\Psi) =\set{v \in V \mid \exists \strab\in\Strab. \ 
\forall \straa\in\Straa.\ \pat(v,\straa,\strab) \in \Psi}.$

\begin{theorem}[Determinacy and complexity]
\label{thrm:determinacy}
The following assertions hold.
\begin{enumerate}
\item For all game graphs $G=((V,E),(V_1,V_2))$, all objectives $\Phi$ for 
player~1 where $\Phi$ is reachability, safety, or parity objectives with three 
priorities, and the complementary objective $\Psi=\Pat \setminus \Phi$ for 
player~2, we have $W_1(\Phi)=V \setminus W_2(\Psi)$; and memoryless winning 
strategies exist for both players from their respective winning set~\cite{EJ91}.

\item The winning set $W_1(\Phi)$ can be computed in linear time ($O(|V|+|E|)$) 
for reachability and safety objectives $\Phi$~\cite{Immerman81,Beeri}; and 
in quadratic time ($O(|V|\cdot |E|)$) for parity objectives 
with three priorities~\cite{Jur00}.

\end{enumerate}
\end{theorem}

\newcommand{\Succ}{\mathsf{Succ}}
\newcommand{\Win}{\mathsf{Win}}

\section{Fair Alternating Simulation}\label{sec:fair_alt}
In this section we will present two definitions of fair alternating simulation,
show their equivalence, present algorithms for solving fair alternating 
simulations, and our algorithms specialized to fair simulation will improve
the bound of the previous algorithm (Theorem~\ref{thrm:fair}).
Similar to fair TSs, a \emph{fair ATS} $\mathcal{K} = \structure{K, F}$ 
consists of an ATS $K$ and a B\"uchi fairness constraint $F$ for $K$. 

To extend the definition of fair simulation to fair alternating simulation we 
consider the notion of strategies for ATSs.
Consider two ATSs 
$K = \structure{\Sigma,W,\whw,A_1,A_2,P_1,P_2,L,\delta}$ and 
$K' = \structure{\Sigma,W',\whw',A_1',A_2',P_1',P_2',L',\delta'}$ and 
the corresponding fair ATSs $\mathcal{K} = \structure{K,F}$ and 
$\mathcal{K}' = \structure{K',F'}$. 
We use the following notations: 
\begin{itemize} \renewcommand{\labelitemi}{$-$}

\item $\tau : (W \times W')^+ \rightarrow A_1$ is a strategy employed by 
Agent~1 in $\mathcal{K}$. The aim of the strategy is to choose transitions 
in $\mathcal{K}$ to make it difficult for Agent~1 in $\mathcal{K}'$ to 
match them. The strategy acts on the past run on both systems.

\item $\tau' : (W \times W')^+ \times A_1 \rightarrow A_1'$ is a strategy 
employed by Agent~1 in $\mathcal{K}'$. The aim of this strategy is to match 
actions in $\mathcal{K}'$ to those made by Agent~1 in $\mathcal{K}$. 
The strategy acts on the past run on both the systems, as well as the action 
chosen by Agent~1 in $\mathcal{K}$.

\item $\xi' : (W \times W')^+ \times A_1 \times A_1' \rightarrow A_2'$ is a 
strategy employed by Agent~2 in $\mathcal{K}'$. The aim of this strategy  is 
to choose actions in $\mathcal{K}'$ to make it difficult for Agent~2 to match 
them in $\mathcal{K}$.  
The strategy acts on the past run of both the systems, as well as the actions 
chosen by Agent~1 in $\mathcal{K}$ and $\mathcal{K}'$.

\item $\xi : (W \times W')^+ \times A_1 \times A_1' \times A_2' \rightarrow A_2$
is a strategy employed by Agent~2 in $\mathcal{K}$. 
Intuitively, the aim of this strategy of Agent~2 is to choose actions in 
$\mathcal{K}$ to show that Agent~1 is not as powerful in $\mathcal{K}$ as in 
$\mathcal{K}'$, i.e., in some sense the strategy of Agent~2 will witness that 
the strategy of Agent~1 in $\mathcal{K}$ does not satisfy certain desired 
property. 
The strategy acts on the past run of both the systems, as well as the actions 
chosen by Agent~1 in $\mathcal{K}$ and both the agents in $\mathcal{K}'$.

\item $\rho(w,w',\tau,\tau',\xi,\xi')$ is the run that emerges in 
$\mathcal{K}$ if the game starts with $\mathcal{K}$ on state $w$, 
$\mathcal{K}'$ on state $w'$ and the agents employ strategies $\tau$, $\tau'$, $\xi$ and $\xi'$ as described above, and  
$\rho'(w,w',\tau,\tau',\xi,\xi')$ is the corresponding run that emerges in 
$\mathcal{K}'$.
\end{itemize}

\begin{definition}[Weak fair alternating simulation]\label{defn:weak_alt_fair_sim}
A binary relation $S \subseteq W \times W'$ is a \emph{weak fair alternating simulation (WFAS)} 
of $\mathcal{K}$ by $\mathcal{K}'$ if the following two conditions hold 
for all $(w,w') \in W\times W'$:
\begin{enumerate}
\item If $S(w,w')$, then $L(w) = L'(w')$.
\item There exists a strategy $\tau' : (W \times W')^+ \times A_1 \rightarrow A_1'$ for Agent~1 in 
$\mathcal{K}'$, such that for all strategies $\tau : (W \times W')^+ \rightarrow A_1$ for Agent~1 
in $\mathcal{K}$, there exists a strategy $\xi : (W \times W')^+ \times A_1 \times A_1' \times A_2 \rightarrow A_2'$ 
for Agent~2 in $\mathcal{K}$, such that for all strategies $\xi' : (W \times W')^+ \times A_1 \times A_1' \rightarrow A_2'$ 
for Agent~2 on $\mathcal{K}'$, if $S(w,w')$ and $\rho(w,w',\tau,\tau',\xi,\xi')$ is a fair $w$-run of $\mathcal{K}$, then
      \begin{itemize} \renewcommand{\labelitemi}{$-$}
      \item $\rho'(w,w',\tau,\tau',\xi,\xi')$ is a fair $w'$-run of $\mathcal{K}'$; and
      \item $\rho'(w,w',\tau,\tau',\xi,\xi')$ $S$-matches $\rho(w,w',\tau,\tau',\xi,\xi')$.
      \end{itemize}
  \end{enumerate}
We denote by $\preceq_{\fairalt}^{\weak}$ the maximum WFAS relation between 
$\mathcal{K}$ and $\mathcal{K'}$.
We say that the fair ATS $\mathcal{K'}$ \emph{weak-fair-alternate simulates} 
the fair ATS $\mathcal{K}$ iff $(\whw,\whw')\in\preceq_{\fairalt}^{\weak}$.
\end{definition}

\begin{definition}[Strong fair alternating simulation]\label{defn:strong_alt_fair_sim}
A binary relation $S \subseteq W \times W'$ is a \emph{strong fair alternating simulation (SFAS)} 
of $\mathcal{K}$ by $\mathcal{K}'$ if the following two conditions hold
for all $(w,w') \in W \times W'$:
\begin{enumerate}
\item If $S(w,w')$, then $L(w) = L'(w')$.
\item There exist strategies $\tau' : (W \times W')^+ \times A_1 \rightarrow A_1'$ for Agent~1 in 
$\mathcal{K}'$ and $\xi : (W \times W')^+ \times A_1 \times A_1' \times A_2 \rightarrow A_2'$ 
for Agent~2 in $\mathcal{K}$,  such that for all strategies $\tau : (W \times W')^+ \rightarrow A_1$ for Agent~1 
in $\mathcal{K}$ and  $\xi' : (W \times W')^+ \times A_1 \times A_1' \rightarrow A_2'$ 
for Agent~2 on $\mathcal{K}'$, if $S(w,w')$ and $\rho(w,w',\tau,\tau',\xi,\xi')$ is a fair $w$-run of $\mathcal{K}$, then
      \begin{itemize} \renewcommand{\labelitemi}{$-$}
      \item $\rho'(w,w',\tau,\tau',\xi,\xi')$ is a fair $w'$-run of $\mathcal{K}'$; and
      \item $\rho'(w,w',\tau,\tau',\xi,\xi')$ $S$-matches $\rho(w,w',\tau,\tau',\xi,\xi')$.
      \end{itemize}
  \end{enumerate}
We denote by $\preceq_{\fairalt}^{\strong}$ the maximum SFAS relation between 
$\mathcal{K}$ and $\mathcal{K'}$.
We say that the fair ATS $\mathcal{K'}$ \emph{strong-fair-alternate simulates} 
the fair ATS $\mathcal{K}$ iff $(\whw,\whw')\in\preceq_{\fairalt}^{\strong}$.
\end{definition}

The difference in the definitions of weak and strong alternating fair 
simulation is in the order of the quantifiers in the strategies.
In the weak version the quantifier order is exists forall exists forall, 
whereas in the strong version the order is exists exists forall forall. 
Thus strong fair alternating simulation implies weak fair alternating 
simulation.
We will show that both the definitions coincide and present algorithms to 
compute the maximum fair alternating simulation.
Also observe that both the weak and strong version coincide with fair 
simulation for TSs.
We will present a reduction of weak and strong fair alternating simulation 
problem to games with parity objectives with three priorities.
We now present a few notations related to the reduction.

\smallskip\noindent{\bf Successor sets.} 
Given an ATS $K$, for a state $w$ and an action $a \in P_1(w)$, let 
$\Succ(w,a)=\set{w' \mid \exists b \in P_2(w) \text{ such that } w'=\delta(w,a,b)}$
denote the possible successors of $w$ given an action $a$ of Agent~1 (i.e.,
successor set of $w$ and $a$).
Let $\Succ(K)=\set{\Succ(w,a) \mid w\in W, a \in P_1(w)}$ denote the set of 
all possible successor sets. 
Note that $|\Succ(K)| \leq |W|\cdot |A_1|$.

\smallskip\noindent{\bf Game construction.}
Let $K=\structure{\Sigma,W,\whw,A_1,A_2,P_1,P_2,L,\delta}$ 
and $K'=\structure{\Sigma,W',\whw',A_1',A_2',P_1',P_2',L',\delta'}$ be two ATSs, 
and let $\mathcal{K} = \structure{K,F}$ and $\mathcal{K}' = \structure{K',F'}$ be the 
two corresponding fair ATSs. 
We will construct a game graph $G=((V,E),(V_1,V_2))$ with a parity objective.
Before the construction we assume that from every 
state $w \in K$ there is an Agent~1 strategy to ensure fairness in $K$. 
The assumption is without loss of generality because
if there is no such strategy from $w$, then trivially all states $w'$ with 
same label as $w$ simulate $w$ (as Agent~2 can falsify the fairness from $w$).
The states in $K$ from which fairness cannot be ensured can be identified with 
a quadratic time pre-processing step in $K$ (solving B\"uchi games), and 
hence we assume that in all remaining states in $K$ fairness can be ensured.
The game construction is as follows:
\begin{itemize} \renewcommand{\labelitemi}{$-$}
\item \emph{Player~1 vertices:} 
$V_1 = \set{\structure{w,w'} \mid w \in W, w' \in W' \text{ such that } L(w) = L'(w')} 
\cup \big(\Succ(K) \times \Succ(K')\big) \cup \set{\frownie}$
\item \emph{Player~2 vertices:} $V_2 = \Succ(K) \times W' \times \set{\#, \$}$

\item \emph{Edges.} We specify the edges as the following union: $E = E_1 \cup E_2 \cup E_3 \cup E_4^1 \cup E_4^2 \cup E_5$
        \begin{eqnarray*}
          E_1 &=&   \set{ (\structure{w,w'}, \structure{\Succ(w,a),w',\#}) \mid \structure{w,w'}\in V_1, a \in P_1(w) } \\[1ex]
          E_2 &=&   \set{ (\structure{T,w',\#}, \structure{T,\Succ(w',a')}) \mid \structure{T,w',\#} \in V_2, a' \in P_1'(w') }\\[1ex]
          E_3 &=&   \set{ (\structure{T,T'}, \structure{T,r',\$}) \mid \structure{T,T'} \in V_1, r' \in T' }\\[1ex]
          E_4^1 &=& \set{ (\structure{T,r',\$}, \structure{r,r'}) \mid \structure{T,r',\$} \in V_2, r \in T, L(r) = L'(r') }\\[1ex]
          E_4^2 &=& \set{ (\structure{T,r',\$}, \frownie) \mid \structure{T,r',\$} \in V_2 \text{ such that } \forall r \in T \cdot L(r) \neq L'(r') }\\[1ex]
          E_5 &=&   \set{ (\frownie, \frownie) }
        \end{eqnarray*}
  \end{itemize}
The intuitive description of the game graph is as follows: 
(i)~the player~1 vertices are either state pairs $\structure{w,w'}$ with same label, or 
pairs $\structure{T,T'}$ of successor sets, or a state $\frownie$; and 
(ii)~the player~2 vertices are tuples $\structure{T,w',\bowtie}$ where $T$ is a 
successor set in $\Succ(K)$, $w'$ a state in $K'$ and 
$\bowtie \in \set{\#,\$}$. 
The edges are described as follows: 
(i)~$E_1$ describes that in vertices $\structure{w,w'}$ player~1 can choose an 
action $a \in P_1(w)$, and then the next vertex is the  player~2 vertex
$\structure{\Succ(w,a),w',\#}$;
(ii)~$E_2$ describes that in vertices $\structure{T,w',\#}$ player~2 can choose an
action $a' \in P_1(w')$ and then the next vertex is $\structure{T,\Succ(w',a')}$;
(iii)~$E_3$ describes that in states $\structure{T,T'}$ player~1 can choose a state 
$r'\in T'$ (which intuitively corresponds to an action $b' \in P_2'(w')$) 
and then the next vertex is $\structure{T,r',\$}$;
(iv)~the edges $E_4^1 \cup E_4^2$ describes that in states $\structure{T,r',\$}$ 
player~2 can either choose a state $r\in T$ that matches the label of $r'$ and then 
the next vertex is the player~1 vertex $\structure{r,r'}$ (edges
$E_4^1$) or if there is no match, then the next vertex is $\frownie$; and 
(v)~finally $E_5$ specifies that the vertex $\frownie$ is an absorbing (sink) 
vertex with only self-loop.
The three-priority parity objective $\Phi^*$ for player~2 with the priority function 
$p$ is specified as follows: 
for vertices $v \in (W \times F')\cap V_1$ we have $p(v)=0$; 
for vertices $v \in ((F \times W' \setminus W \times F')\cap V_1) 
\cup \set{\frownie}$ we have $p(v)=1$;
and all other vertices have priority~2.

\smallskip\noindent{\em Plays and runs.}
 Every $\structure{w,w'}$-play on the game (plays that start from vertex $\structure{w,w'}$) induces 
runs on the structures $\mathcal{K}$ and $\mathcal{K}'$ as follows :
\begin{itemize} \renewcommand{\labelitemi}{$-$}
\item $\structure{w,w'}$, $\structure{T_0,w',\#}$, $\structure{T_0,T_0'}$, $\structure{T_0,w_1',\$}$, $\structure{w_1,w_1'}$, 
$\structure{T_1,w_1',\#}$, $\structure{T_1,T_1'}$, $\structure{T_1,w_2',\$}$, $\structure{w_2,w_2'}$, $\dots$ corresponds to runs 
$\ov{w} = w,w_1,w_2\dots$ and $\ov{w}' = w',w_1',w_2'\dots$.
\item $\structure{w,w'}$, $\structure{T_0,w',\#}$, $\structure{T_0,T_0'}$, $\structure{T_0,w_1',\$}$, $\structure{w_1,w_1'}$, $\structure{T_1,w_1',\#}$, $\structure{T_1,T_1'}$, $\structure{T_1,w_2',\$}$, $\structure{w_2,w_2'}$, $\dots$, $\structure{w_{n-1},w_{n-1}'}$, $\structure{T_{n-1},w_{n-1}',\#}$, $\structure{T_{n-1},T_{n-1}'}$, $\structure{T_{n-1},w_n',\$}$, $\frownie$, $\frownie$, $\dots$ corresponds to finite runs $\bar{w} = w, w_1, w_2 \dots w_{n-1}, w_{n}$ and $\bar{w}' = w', w_1', w_2' \dots w_{n}'$, for some $w_n\in T_{n-1}$.
\end{itemize}

\begin{lemma} \label{lemma:alt_fair_correspondence}
Consider a play $\overline{\structure{w,w'}}$ = $\structure{w,w'}$, $\structure{T_0,w',\#}$, $\structure{T_0,T_0'}$, $\structure{T_0,w_1',\$}$, $\structure{w_1,w_1'}$, $\dots$ on 
the parity game.
Then the following assertions hold:
\begin{enumerate}
\item If the play satisfies the parity objective, then the corresponding runs 
$\ov{w} =w_0,w_1, w_2\dots$ in  $\mathcal{K}$ (where $w_0=w$) and 
$\ov{w'} = w_0',w_1',w_2'\dots$ in $\mathcal{K}'$ (where $w_0'=w'$)  
satisfy that if $\ov{w}$ is fair, then $\ov{w}'$ is fair and for all $i \geq 0$ we have 
$L(w_i) = L'(w_i')$ .

\item If the play does not satisfy the parity objective, then 
(i)~if the vertex $\frownie$ is not reached, then the corresponding runs 
$\ov{w} = w,w_1, w_2\dots$ in  $\mathcal{K}$ and 
$\ov{w'} = w',w_1',w_2'\dots$ in $\mathcal{K}'$ 
satisfy that $\ov{w}$ is fair and $\ov{w}'$ is not fair;
(ii)~if the vertex $\frownie$ is reached, then  for the corresponding finite 
runs $\ov{w} = w, w_1, w_2 \dots w_{n}$ and $\ov{w}' = w', w_1', w_2' \dots w_{n}'$ we have 
that $w_{n}'$ does not match $w_{n}$ (i.e., $L(w_n)\neq L'(w_n')$).
\end{enumerate}
\end{lemma}
\begin{proof} 
We prove both the items below:
\begin{enumerate}

\item If the parity objective is satisfied, it follows that the vertex  
$\frownie$ is never reached. 
By construction of the game, vertices of the form $\structure{w,w'}$ satisfy 
that $L(w)=L'(w')$, and it follows that for all $i \geq 0$ we have 
$L(w_i)=L'(w_i')$.
Moreover, as the parity objective is satisfied, it follows that if in $K$, 
states in $F$ are visited infinitely often, then in $K'$, states in $F'$ 
must be visited infinitely often,
(as otherwise priority~1 vertices will be visited infinitely often and 
priority~0 vertices only finitely often). 
This completes the proof of the first item.

\item If the parity objective is not satisfied, and the vertex $\frownie$ 
is never reached, it follows that priority~1 vertices in 
$(F\times W' \setminus W\times F')\cap V_1$ 
are visited infinitely often (hence $F$ is visited infinitely often in $K$) and 
priority~0 vertices ($(W \times F')\cap V_1$) are visited finitely often 
(hence $F'$ is visited finitely often in $K'$).
Thus we have a fair run in $K$, but the run in $K'$ is not fair.
If the $\frownie$ vertex is reached, then by construction it follows that 
$L(w_n)\neq L'(w_n')$. 

\end{enumerate}
The desired result follows.
\hfill\qed
\end{proof}

\smallskip\noindent{\em Consequence of Lemma~\ref{lemma:alt_fair_correspondence}.}
We have the following consequence of the lemma.
If a play satisfies the parity objective, then the corresponding runs satisfy 
that if we have a fair run in $K$, then the run in $K'$ is both fair and matches
the run in $K$.
If the play does not satisfy the parity objective, then we have two cases:
(i)~the run in $K$ is fair, but the run in $K'$ is not fair; or 
(ii)~the run in $K'$ does not match the run in $K$, and since we assume 
that from every state in $K$ fairness can be ensured, it follows that once
we have the finite non-matching run, we can construct a fair run in $K$
that is not matched in $K'$.
Thus if the play does not satisfy the parity objective, then in
both cases we have a fair run in $K$ and the run in $K'$ is either not fair 
or does not match the run in $K$.

\begin{proposition} \label{prop:alt_fair_sim_correctness}
Let $\Win_2= \set{(w_1,w_2) \mid  \structure{w_1,w_2} \in V_1, \structure{w_1,w_2} \in W_2(\Phi^*), \text{i.e.,there is a winning state for player 2}}$.
Then we have 
\[
\Win_2= \preceq_{\fairalt}^{\weak} = \preceq_{\fairalt}^{\strong}.
\]
\end{proposition}
\begin{proof}
We first note that by definition we have 
$\preceq_{\fairalt}^{\strong} \subseteq \preceq_{\fairalt}^{\weak}$. 
Hence to prove the result it suffices to show the following inclusions: 
(i)~$\Win_2 \subseteq \preceq_{\fairalt}^{\strong}$ 
and (ii)~$\preceq_{\fairalt}^{\weak} \subseteq \Win_2$.
We prove the inclusions below:

\begin{enumerate}
\item \emph{(First inclusion: $\Win_2 \subseteq \preceq_{\fairalt}^{\strong}$).}  
We need to show that $\Win_2$ is a strong fair alternating simulation. 
Let $(w,w') \in \Win_2$, then $\structure{w,w'} \in V_1$ and by construction 
of the game we have $L(w) = L'(w')$. 
Hence we need to show that there exist  strategies $\tau'$ and $\xi$, 
such that for all strategies $\tau$ and $\xi'$, 
we have that if $\rho(w,w',\tau,\tau',\xi,\xi')$ is a fair $w$-run in $\mathcal{K}$, 
then $\rho'(w,w',\tau,\tau',\xi,\xi')$ is a fair $w'$-run in $\mathcal{K}'$ 
and $\rho'(w,w',\tau,\tau',\xi,\xi')$ $\Win_2$-matches $\rho(w,w',\tau,\tau',\xi,\xi')$.
      
Since $\structure{w,w'}$ is a winning vertex for player 2, there exists 
a memoryless winning strategy  $\strab^m$ for player 2, 
which will ensure that all plays starting from $\structure{w,w'}$ and 
consistent with $\strab^m$ will satisfy the parity objective.
Note that the strategy $\strab^m$ specifies the next vertices 
for vertices in $\Succ(K) \times W' \times \set{\#,\$}$.
Using $\strab^m$ we can construct the required witness strategies $\tau'$ and $\xi$ for 
strong fair alternating simulation as follows: 
\[
\tau'[\structure{w,w'},\structure{w_1,w_1'}, \dots,\structure{w_{n-1},w_{n-1}'},a]  =  a' \in P_1'(w_{n-1}') 
\]
such that  $\Succ(w_{n-1}',a') = \Pi_{(2)}(\strab^m[\structure{T_{n-1},w_{n-1}',\#}])$, where 
$T_{n-1}=\Succ(w_{n-1},a)$; and
\[
\xi[\structure{w,w'},\structure{w_1,w_1'},\dots,\structure{w_{n-1},w_{n-1}'},a,a',b'] = b \in P_2(w_{n-1}) 
\]
such that $\delta'(w_{n-1},a,b) = \Pi_{(1)}(\strab^m[\structure{T_{n-1},w_{n}',\$}])$, where 
$T_{n-1}=\Succ(w_{n-1},a)$ and $w_n'=\delta'(w_{n-1}',a',b')$;
($\Pi$ is the projection operator, that is, $\Pi_{(k)}(x_1, x_2, \dots, x_n) = x_k$). 
Note that if the game reaches the vertex $\frownie$, then the objective $\Phi^*$ for player~2 is 
violated and player 1 would win. 
Hence,  since $\strab^m$ is a winning strategy for player 2, it ensures that the play never reaches 
$\frownie$. Hence, the outcome of $\strab^m$ on which the projection operator acts always lies 
in $V_1 \setminus \set{\frownie}$, and hence is a \emph{2-tuple}.
Consider a $\structure{w,w'}$-play consistent with the strategy $\strab^m$, where 
$\structure{w,w'}$ is in $\Win_2$.
As described earlier, the $\structure{w,w'}$-play of the parity game defines two runs: a 
$w$-run, $\ov{w} = w,w_1,w_2,\dots$ in $K$ and a $w'$-run $\ov{w}' = w',w_1',w_2'\dots$
in $K'$. 
Since $\structure{w,w'}$ is a winning state for player 2, all successor states 
$\structure{w_k,w_k',\#}$ must also be winning states for player 2. 
Hence $(w_k,w_k') \in \Win_2$ for all $k \in \mathbb{N}$, and it follows
that the run $\ov{w}'$ in $K'$ $\Win_2$-matches $\ov{w}$ in $K$.
Since $\strab^m$ ensures the parity objective $\Phi^*$ (all plays consistent with 
$\strab^m$ satisfy $\Phi^*$), it follows from 
Lemma~\ref{lemma:alt_fair_correspondence} that for all strategies $\tau$ and $\xi$ if 
$\rho(w,w',\tau,\tau',\xi,\xi')$ is a fair run on $\mathcal{K}$ (visits $F$ infinitely 
often), 
then $\rho'(w,w',\tau,\tau',\xi,\xi')$ is a fair run on $\mathcal{K}'$ 
(visits $F'$ infinitely often).
Hence we have the desired first inclusion: $\Win_2 \subseteq \preceq_{\fairalt}^{\strong}$.

\item \emph{(Second inclusion: $\preceq_{\fairalt}^{\weak}\subseteq \Win_2$).}  
We need to show that if $(w,w') \in \preceq_{\fairalt}^{\weak}$, 
then $\structure{w,w'}$ is a winning vertex for player 2 in the game, 
that is, there exists a strategy $\strab$ for player 2 such that against all strategies of 
player~1 the parity objective $\Phi^*$ is satisfied.
By determinacy of parity games on graphs, instead of a winning strategy for 
player~2 it suffices to show that against every strategy $\straa$ of player~1 there 
is a strategy $\strab$ (dependent on $\straa$) for player~2 to ensure winning
against $\straa$.
Since $(w,w') \in \preceq_{\fairalt}^{\weak}$ we have 
(i)~$L(w) = L'(w')$ and 
(ii)~there exist a strategy $\tau'$, such that for all strategies $\tau$, 
there exists a strategy $\xi$, such that for all strategies $\xi'$, 
if $\rho(w,w',\tau,\tau',\xi,\xi')$ is a fair $w$-run $\ov{w}$ in $\mathcal{K}$, 
then $\rho'(w,w',\tau,\tau',\xi,\xi')$ is a fair $w'$-run $\ov{w}'$ in 
$\mathcal{K}'$ and $\rho'(w,w',\tau,\tau',\xi,\xi')$ 
$\preceq_{\fairalt}^{\weak}$-matches $\rho(w,w',\tau,\tau',\xi,\xi')$. 
Consider a strategy $\straa$ for player~1, and let $\tau$ and $\xi'$ be 
the corresponding strategies obtained from $\straa$.
We construct the desired strategy $\strab$ from $\tau'$ and $\xi$ as follows:
\[
\strab[\structure{w,w'}, \dots, \structure{w_{n-1},w_{n-1}'}, \structure{T_{n-1},w_{n-1}',\#}]
 = \structure{T_{n-1},\Succ(w_{n-1}',\tau'[\structure{w,w'},\structure{w_1,w_1'},\dots,\structure{w_{n-1},w_{n-1}'},a])}; 
\]
where $a$ is such that $T_{n-1}=\Succ(w_{n-1},a)$, and
\[
\strab[\structure{w,w'}, \dots, \structure{T_{n-1},T_{n-1}'}, \structure{T_{n-1},w_{n}',\$}]
= \structure{\delta(w_{n-1},a,\xi[\structure{w,w'},\structure{w_1,w_1'},\dots,\structure{w_{n-1},w_{n-1}'},a,a',b']),w_n'}
\]
where $a$ is such that $T_{n-1}=\Succ(w_{n-1},a)$, and $a'$ such that 
$T_{n-1}'=\Succ(w_{n-1}',a')$ and $b'$ such that $\delta'(w_{n-1}',a',b')=w_n'$.
We have $\rho'(w,w',\tau,\tau',\xi,\xi')$ 
$\preceq_{\fairalt}^{\weak}$-matches $\rho(w,w',\tau,\tau',\xi,\xi')$, 
we have $L(w_k) = L'(w_k')$ for all $k \in \mathbb{N}$.
It follows that given the strategy $\straa$ and $\strab$ the vertex $\frownie$ is not reached.
Since strategies $\tau'$ and $\xi$ form a \emph{witness} to 
weak fair alternating simulation, it follows that if 
the run $\rho(w,w',\tau,\tau',\xi,\xi')$ is fair, 
then $\rho'(w,w',\tau,\tau',\xi,\xi')$ is fair, and then by 
Lemma \ref{lemma:alt_fair_correspondence} it follows 
that the play given $\straa$ and $\strab$ satisfies the parity objective.
It follows that against the strategy $\straa$ of player~1, 
the strategy $\strab$ is winning for player~2.
Thus it follows that we have $\preceq_{\fairalt}^{\weak}\subseteq \Win_2$.  
\end{enumerate}
The desired result follows.
\hfill\qed
\end{proof}

\begin{lemma}\label{lem_fair_alt_size}
For the game graph constructed for fair alternating simulation we have 
$|V_1| + |V_2| \leq  O(|W|\cdot |W'|\cdot |A_1|\cdot |A_1'|)$; and 
$|E| \leq O(|W|\cdot |W'|\cdot |A_1|\cdot (|A_1'|\cdot |A_2'| + |A_2|))$.
\end{lemma}
\begin{proof}
We have $|\Succ(K)| \leq |W| \cdot |A_1|$ and $|\Succ(K')|\leq |W'| \cdot |A_1'|$. 
Hence we have
\[
|V_1| \le |W \times W'| + |\Succ(K)\times \Succ(K')| + 1 \le 
|W|\cdot |W'| + (|W|\cdot |A_1|)\cdot (|W'|\cdot |A_1'|) + 1 \le 
 O(|W|\cdot |W'|\cdot |A_1|\cdot |A_1'|);
\]
and
\[
|V_2| = 2\cdot |\Succ(K)\times W'| \le 2\cdot (|W|\cdot |A_1|)\cdot |W'|
\]
Thus we have the result for the vertex size.
We now obtain the bound on edges.
We have $|E| = |E_1| + |E_2| + |E_3| + |E_4^1| + |E_4^2| + |E_5|$, and we obtain bound for them
below:
\[
|E_1| \le \sum_{w' \in W'} \sum_{w \in W} |P_1(w)|\le |W'|\cdot |W|\cdot |A_1|
\]
\[
|E_2| = \sum_{T \in \Succ(K)} \sum_{w' \in W'} |P_1'(w')|
\le |\Succ(K)|\cdot |W'|\cdot |A_1'|
= |W|\cdot |W'|\cdot |A_1|\cdot |A_1'|
\]
\[
|E_3| = \sum_{T \in \Succ(K)} \sum_{T' \in \Succ(K')} |T'|
\le |\Succ(K)|\cdot |\Succ(K')|\cdot |A_2'| \le 
|W|\cdot |W'|\cdot |A_1|\cdot |A_1'|\cdot |A_2'|
\]
where for the first inequality above we used the fact that $|T'| \le |A_2'|$;
\[
|E_4^1| = \sum_{r' \in W'} \sum_{T \in \Succ(K)} |T|
\le |W'|\cdot |\Succ(K)|\cdot |A_2|
\le |W'|\cdot |W|\cdot |A_1|\cdot |A_2|
\]
where for the first inequality above we used that $|T| \le |A_2|$; 
\[
|E_4^2| \le \sum_{r' \in W'} \sum_{T \in \Succ(K)} 1 
\le |W'|\cdot |\Succ(K)| \le |W'|\cdot |W|\cdot |A_1|;
\]
and finally $|E_5|=1$.
Hence we have $|E| = O(|W|\cdot |W'|\cdot |A_1|\cdot (|A_1'|\cdot |A_2'| + |A_2|))$.
\hfill\qed
\end{proof}  

The above lemma bounds the size of the game, and we require that 
the game graph can be constructed in time quadratic in the size of the 
game graph and in the following section we will present a more efficient
(than quadratic) construction of the game graph.  
Proposition~\ref{prop:alt_fair_sim_correctness},
along with the complexity to solve parity games with three priorities gives
us the following theorem. 
The result for fair simulation follows as a special case and the details 
are presented in the technical details appendix.

\begin{theorem}\label{thrm:alt_fair}
We have $\preceq_{\fairalt}^{\weak} = \preceq_{\fairalt}^{\strong}$, 
the relation $\preceq_{\fairalt}^{\strong}$ can be computed in time 
$O(|W|^2\cdot |W'|^2\cdot |A_1|^2\cdot |A_1'|\cdot (|A_1'|\cdot|A_2'| + |A_2|))$
for two fair ATSs $\mathcal{K}$ and $\mathcal{K}'$.
The fair simulation relation $\preceq_{\fair}$ can be computed in time 
$O(|W|\cdot |W'| \cdot (|W|\cdot |R'| + |W'|\cdot |R|))$ 
for two fair TSs $\mathcal{K}$ and $\mathcal{K}'$.
\end{theorem}

\begin{remark}
We consider the complexity of fair simulation, and let 
$n=|W|=|W'|$ and $m=|R|=|R'|$.
The previous algorithm of~\cite{HKR97} requires time $O(n^6)$ and 
our algorithm requires time $O(n^3 \cdot m)$.
Since $m$ is at most $n^2$, our algorithm takes in worst case 
time $O(n^5)$ and in most practical cases we have $m=O(n)$ and 
then our algorithm requires $O(n^4)$ time as compared to the 
previous known $O(n^6)$ algorithm.
\end{remark}

\begin{comment}
\subsubsection{Reduction Complexity}
In a similar way as done in \cite{paper:alternating_simulation}, the reduction can be done in $\mathcal{O}(|W|*\sum\limits_{w \in W} |P_1(w)| + |W'|*\sum\limits_{w' \in W'} |P'_1(w')|) + \mathcal{O}(|V|+|E|)$. Since the reduction complexity time is less than that of find the set of winning states, the reduction complexity is insignificant.\\

Thus the overall complexity of the algorithm is $\mathcal{O}((|V_1|+|V_2|)*|E|) = \mathcal{O}(|W|*|W'|*|A_1|*|A_1'|*|W|*|W'|*|A_1|*(|A_1'|*|A_2'| + |A_2|)) = \mathcal{O}(|W|^2*|W'|^2*|A_1|^2*|A_1'|*(|A_1'|*|A_2'| + |A_2|))$.
\end{comment}

\newcommand{\prev}{\mathit{prev}}
\newcommand{\BT}{\mathsf{BT}}
\newcommand{\Lf}{\mathsf{Lf}}
\newcommand{\Ar}{\mathsf{Ar}}
\newcommand{\GAr}{\mathsf{GAr}}

\newcommand{\simstr}{\mathsf{sim}}
\newcommand{\countstr}{\mathsf{count}}
\newcommand{\remove}{\mathsf{remove}}
\newcommand{\pre}{\mathsf{Pre}}
\newcommand{\post}{\mathsf{Post}}
\newcommand{\prevsim}{\mathsf{prevsim}}
\newcommand{\inv}{\mathsf{inv}}

\section{Alternating Simulation}
In this section we will present two algorithms to compute the maximum 
alternating simulation relation for two ATSs $K$ and $K'$. 
The first algorithm for the problem was presented in~\cite{AHKV98} and 
we refer to the algorithm as the basic algorithm.
The basic algorithm iteratively consideres pairs of states and examines 
if they are already witnessed to be not in the alternating simulation 
relation, removes them and continues until a fix-point is reached. 
The algorithm is described as Algorithm~\ref{algo:basic} (see Theorem~3
of~\cite{AHKV98}).
The correctness of the basic algorithm was shown in~\cite{AHKV98}, and 
the time complexity of the algorithm is $O(|W|^2\cdot |W'|^2\cdot |A_1|\cdot |A_1'|\cdot |A_2|\cdot |A_2'|)$:
(i)~time take by \emph{If} condition is $O(|A_1|\cdot |A_1'|\cdot |A_2|\cdot |A_2'|)$;
(ii)~time taken by the nested \emph{For} loops is $O(|W|\cdot |W'|)$; and
(iii)~the maximum number of iterations of the \emph{While} loop is $O(|W|\cdot |W'|)$.

\begin{algorithm*}[t]
\caption{Basic Algorithm}
\label{algo:basic}
{
\begin{tabbing}
aaa \= aaa \= aaa \= aaa \= aaa \= aaa \= aaa \= aaa \kill
\> {\bf Input:} $K= (\Sigma,W,\whw,A_1,A_2,P_1,P_2,L,\delta)$, $K'= (\Sigma,W',\whw',A_1',A_2',P_1',P_2',L',\delta')$. \\
\>   {\bf Output:} $\preceq_{\altsim}$. \\ 

\> 0. $\preceq^{\prev} \gets W \times W'$; $\preceq \gets \set{ (w,w') \mid w \in W, w' \in W', L(w) = L'(w') }$; \\ 
\> 1. {\bf while ($\preceq^{\prev} \ne \preceq$) } \\ 
\>\> 1.1 $\ \ \preceq^{\prev} \gets \preceq$ \\
\>\> 1.2 \ \ {\bf for all   $w \in W$, $w' \in W'$} \\
\>\>\>\>  {\bf if  (}$w \preceq^{\prev} w'$ and 
$\exists a \in P_1(w) \cdot \forall a' \in P'_1(w') \cdot \exists b' \in P'_2(w') \cdot \forall b \in P_2(w) \cdot \delta(w,a,b) \npreceq^{\prev} \delta(w',a',b')${\bf ), then} \\
\>\>\>\>\>  $\preceq \gets \preceq \setminus \set{(w,w')}$ \\
\> 2. {\bf return} $\preceq$.  
\end{tabbing}
}
\end{algorithm*}

\begin{comment}
\begin{algorithm}
\caption{Basic Algorithm}
\label{algo:basic}
\begin{algorithmic}
\renewcommand{\algorithmicrequire}{\textbf{Input:}}
\renewcommand{\algorithmicensure}{\textbf{Output:}}
  \REQUIRE $K= (\Sigma,W,\whw,A_1,A_2,P_1,P_2,L,\delta)$, $K'= (\Sigma,W',\whw',A_1',A_2',P_1',P_2',L',\delta')$
  \ENSURE $\preceq_{\altsim}$ 
  \STATE $\preceq^{\prev} \gets W \times W'$; $\preceq \gets \set{ (w,w') \mid w \in W, w' \in W', L(w) = L'(w') }$
  \WHILE {$\preceq^{\prev} \ne \preceq$}
  \STATE $\preceq^{\prev} \gets \preceq$
  \FORALL {$w \in W$, $w' \in W'$ }
      \IF {$w \preceq^{\prev} w'$}
        \IF{$ \exists a \in P_1(w) \cdot \forall a' \in P'_1(w') \cdot \exists b' \in P'_2(w') \cdot \forall b \in P_2(w) \cdot \delta(w,a,b) \npreceq^{\prev} \delta(w',a',b')$}
           \STATE $\preceq \gets \preceq \setminus \set{(w,w')}$
        \ENDIF
      \ENDIF
    \ENDFOR
  \ENDWHILE
 \STATE {\bf RETURN} $\preceq$
\end{algorithmic}
\end{algorithm}
\end{comment}

\subsection{Improved Algorithm Through Games}\label{sec:alt-sim-1}
In this section we present an improved algorithm for alternating simulation 
by reduction to reachability-safety games.

\smallskip\noindent{\bf Game construction.} 
Given two ATSs $K= (\Sigma,W,\whw,A_1,A_2,P_1,P_2,L,\delta)$ and 
$K'= (\Sigma,W',\whw',A_1',A_2',P_1',P_2',L',\delta')$, we construct a game graph 
$G=((V,E),(V_1,V_2))$ as follows:
\begin{itemize}
\item \emph{Player~1 vertices:} $V_1 = (W \times W') \cup \big(\Succ(K) \times \Succ(K')\big)$;
\item \emph{Player~2 vertices:} $V_2 = \Succ(K)\times W' \times \{\#, \$\}$;
\item \emph{Edges:} The edge set $E$ is specified as the following union: $E = E_1 \cup E_2 \cup E_3 \cup E_4$
        \begin{eqnarray*}
          E_1 &=& \set{ (\structure{w,w'}, \structure{\Succ(w,a),w',\#}) \mid w \in W, w' \in W', a \in P_1(w)} \\
          E_2 &=& \set{ (\structure{T,w',\#}, \structure{T,\Succ(w',a')}) \mid T \in \Succ(K), w' \in W', a' \in P_1'(w')}\\
          E_3 &=& \set{ (\structure{T,T'}, \structure{T,r',\$}) \mid T \in \Succ(K), T' \in \Succ(K'), r' \in T'}\\
          E_4 &=& \set{ (\structure{T,r',\$}, \structure{r,r'}) \mid T \in \Succ(K), r' \in W', r \in T}
        \end{eqnarray*}
\end{itemize}
Let $T = \set{ \structure{w,w'} \mid  L(w) \ne L'(w') }$ be the state pairs that does not match by the 
labeling function, and let $F=V \setminus T$.
The objective for player~1 is to reach $T$ (i.e., $\Reach{T}$) and 
the objective for player~2 is the safety objective $\Safe{F}$.
In the following proposition we establish the connection of the winning set for 
player~2 and $\preceq_{\altsim}$.

\begin{proposition}\label{prop:alt-sim}
Let $\Win_2= \set{(w,w') \mid w \in W, w' \in W', \structure{w,w'} \in W_2(\Safe{F}) \text{ i.e., is a winning vertex for player 2}}$.
Then we have $\Win_2=\preceq_{\altsim}$.
\end{proposition}
\begin{proof}
We prove the result by proving two inclusions: (i)~$\Win_2\subseteq \preceq_{\altsim}$ and (ii)~$\preceq_{\altsim} \subseteq \Win_2$.
\renewcommand{\labelitemi}{$\bullet$}
\begin{enumerate}
\item \emph{(First inclusion: $\Win_2 \subseteq \preceq_{\altsim}$).} 
We show that $\Win_2$ is an alternating simulation relation. 
Let $\structure{w,w'}$ be a winning vertex in $\Win_2$ for player 2. 
Since the set of winning vertices is disjoint from 
$T = \set{ \structure{w,w'} \mid L(w) \ne L'(w') }$, 
we can conclude that $L(w) = L'(w')$. 
Thus, we only need to show that for all $(w,w')\in \Win_2$ we have
\begin{equation*}
\forall a \in P_1(w) \cdot \exists a' \in P_1'(w') \cdot \forall b' \in P_2'(w') \cdot \exists b \in P_2(w) \cdot (\delta(w,a,b),\delta'(w',a',b')) \in \Win_2 
\end{equation*}
We have the following analysis:
\begin{itemize}
\item Since $\structure{w,w'}$ is a player-1 vertex, all transitions of player 1 to $\structure{\Succ(w,a),w',\#}$ must be a winning vertex for player~2 for all $a \in P_1(q)$.
\item Since $\structure{\Succ(w,a),w',\#}$ is a player-2 vertex and is a winning vertex for player~2, there exists a transition, that is, there exists $a' \in P'_1(w')$, 
such that $\structure{\Succ(w,a),\Succ(w',a')}$ is a winning vertex for player 2.
\item Since $\structure{\Succ(w,a),\Succ(w',a')}$ is a player-1 vertex and is a winning vertex for player 2, for all transitions, that is, for all $b' \in P'_2(w')$, 
$\structure{\Succ(w,a),\delta'(w',a',b'),\$}$ is a winning vertex for player 2.
\item Since $\structure{\Succ(w,a),\delta'(q',a',b'),\$}$ is a player-2 vertex and is a winning vertex for player 2, there exists a transition, that is, there exists $b \in P_2(w)$
such that $\structure{\delta(w,a,b),\delta'(w',a',b')}$ is a winning vertex for player 2.
\end{itemize}
It follows that $\Win_2$ is an alternating simulation relation and hence $\Win_2 \subseteq \preceq_{\altsim}$.

\item\emph{(Second inclusion: $\preceq_{\altsim} \subseteq \Win_2$).} 
We need to show that $\structure{w,w'}$ is a winning vertex for player 2, for all $(w,w') \in \preceq_{\altsim}$. 
Since $(w,w') \in \preceq_{\altsim}$, it follows that $L(w)=L'(w')$. Hence $\preceq_{\altsim}$ is disjoint from 
$T =\set{ \structure{w,w'} \mid L(w) \ne L'(w') }$. 
Thus, it suffices to show that starting from $\structure{w,w'}$ the player~2 can force that the game never reaches $T$. 
We know that for all $(w,w') \in \preceq_{\altsim}$ we have 
\begin{equation*}
\forall a \in P_1(w) \cdot \exists a' \in P_1'(w') \cdot \forall b' \in P_2'(w') \cdot \exists b \in P_2(w) \cdot 
(\delta(w,a,b),\delta'(w',a',b')) \in \preceq_{\altsim}
\end{equation*}
Thus, starting from all vertices $\structure{w,w'}$ such that $(w,w') \in \preceq_{\altsim}$ the player 2 can force that the game reaches some 
$\structure{r,r'}$ such that $(r,r') \in \preceq_{\altsim}$, that is, player 2 can force that the game always stays in states in $F=V \setminus T$
(as $\preceq_{\altsim} \cap T = \emptyset$).
Hence $\preceq_{\altsim} \subseteq \Win_2$.
\end{enumerate}
The desired result follows.
\hfill\qed
\end{proof}

The algorithmic analysis will be completed in two steps: (1)~estimating the size 
of the game graph; and (2)~analyzing the complexity to construct the game graph 
from the ATSs.

\begin{lemma}\label{lemm:alt-sim}
For the game graph constructed for alternating simulation, we have 
$|V_1| + |V_2| \leq O(|W|\cdot |W'|\cdot |A_1|\cdot |A_1'|)$ and 
$|E| \leq O(|W|\cdot |W'|\cdot |A_1|\cdot (|A_1'|\cdot |A_2'|+|A_2|))$.
\end{lemma} 
\begin{proof}
We have 
\[
|V_1| = |W \times W'| + |\Succ(K) \times \Succ(K')| \le |W|\cdot |W'| + (|W|\cdot |A_1|)\cdot (|W'|\cdot |A_1'|) 
=O(|W|\cdot |W|' \cdot |A_1| \cdot |A_1'|);
\]
\[
|V_2| = 2\cdot |\Succ(K) \times W'| 
\le 2 \cdot (|W|\cdot |A_1|)\cdot |W'| = 2\cdot |W|\cdot |W'|\cdot |A_1|
\]
The bound for $|V_1| + |V_2|$ follows. We now consider the bound for 
the size of $E$.
We have $|E| = |E_1| + |E_2| + |E_3| + |E_4|$, and we obtain bounds for
them below:
\[
|E_1| = \sum_{w' \in W'} \sum_{w \in W} |P_1(w)| \le |W'|\cdot |W|\cdot |A_1|
\]
\[
|E_2| = \sum_{T \in \Succ(K)} \sum_{w' \in W'} |P_1'(w')| \le |\Succ(K)|\cdot |W'|\cdot |A_1'|
\le |W|\cdot |W'|\cdot |A_1| \cdot |A_1'|
\]
\[
|E_3| = \sum_{T \in \Succ(K)} \sum_{T' \in \Succ(K')} |T'|
\le |\Succ(K)|\cdot |\Succ(K')|\cdot |A_2'| \le |W|\cdot |W'|\cdot |A_1|\cdot |A_1'|\cdot |A_2'|
\]
\[
|E_4| = \sum_{r' \in W'} \sum_{T \in \Succ(K)} |T|
\le |W'|\cdot |\Succ(K)|\cdot |A_2| \le |W'|\cdot |W|\cdot |A_1|\cdot |A_2|
\]
where in the bound for $E_3$ we used $|T'| \le |A_2|$ and in the bound for $E_4$ we used 
$|T| \le |A_2|$.
It follows that $|E| = O(|W|\cdot |W'|\cdot |A_1|\cdot (|A_1'|\cdot |A_2'|+|A_2|))$,
and the desired result follows.
\hfill\qed
\end{proof}

\smallskip\noindent{\bf Game graph construction complexity.} 
We now show that the game graph can be constructed in time linear 
in the size of the game graph.
The data strucutre for the game graph is as follows: we map every 
vertex in $V_1 \cup V_2$ to a unique integer, and construct the list 
of edges. 
Given this data structure for the game graph, the winning sets for 
reachability and safety objectives can be computed in linear time~\cite{Beeri,Immerman81}.
We now present the details of the construction of the game graph 
data structure.

\smallskip\noindent{\em Basic requirements.} We start with some basic facts.
For two sets $A$ and $B$, if we have  two bijective functions 
$f_A : A \leftrightarrow \set{ 0, \dots , |A|-1 }$ and 
$f_B : B \leftrightarrow \set{ 0, \dots , |B|-1 }$, 
then we can assign a unique integer to elements of $A \times B$ in 
time $O(|A|\cdot |B|)$. 
Since it is easy to construct bijective functions for $W$ and $W'$, 
we need to construct such bijective functions for $\Succ(K)$ and 
$\Succ(K')$ to ensure that every vertex has a unique integer.
We will present data structure that would achieve the following: 
(i)~construct bijective function $f_K:\Succ(K)\leftrightarrow  \set{0,\ldots,|\Succ(K)|-1}$; 
(ii)~construct function $h_K: W\times A_1 \to \set{0,\ldots, |\Succ(K)|-1}$ such that 
for all $w \in W$ and $a \in P_1(w)$ we have $h_K((w,a))=f_K(\Succ(w,a))$, i.e., it gives
the unique number for the successor set of $w$ and action $a$;
(iii)~construct function $g_K:\set{0,1,\ldots, |\Succ(K)|-1} \to 2^W$ such that for all 
$T \in \Succ(K)$ we have $g_K(f_K(T))$ is the list of states in $T$.
We will construct the same for $K'$, and also ensure that for all $T$ we 
compute $g_K(f_K(T))$ in time proportional to the size of $T$.
We first argue how the above functions are sufficient to construct every edge in 
constant time: 
(i)~edges in $E_1$ can be constructed by considering state pairs $\structure{w,w'}$, 
actions $a \in P_1(w)$, and with the function $h_K((w,a))$ we add the required edge,
and the result for edges $E_2$ is similar with the function $h_{K'}$; 
(ii)~edges in $E_3$ and $E_4$ are generated using the function $g_K$ that gives 
the list of states for $g_K(f_K(T))$ in time proportional to the size of $T$.
Hence every edge can be generated in constant time, given the functions, 
and it follows that with the above functions the game construction is achieved
in linear time. 
We now present the data structure to support the above functions.

\smallskip\noindent{\em Binary tree data structure.} 
Observe that $\Succ(K)$ is a set such that each element is a successor set 
(i.e., elements are set of states).
Without efficient data structure the requirements for the functions 
$f_K,h_K,$ and $g_K$ cannot be achieved.
The data structure we use is a \emph{binary tree data structure}.
We assume that states in $W$ are uniquely numbered from $1$ to $|W|$
Consider a binary tree, such that every leaf has depth $|W|$, i.e., the length of the path 
from root to a leaf is $|W|$. 
Each path from the root to a leaf represents a set --- 
every path consists of a $|W|$ length sequence of \emph{left} and \emph{right} choices. 
Consider a path $\pi$ in the binary tree, and the path $\pi$ represents a subset $W_\pi$ of $W$ 
as follows: if the $i$-th step of $\pi$ is \emph{left}, then $w_i \notin W_\pi$, 
if the $i$-th step is \emph{right}, then $w_i \in W_\pi$. 
Thus, $\Succ(K)$ is the collection of all sets represented by paths (from root to leaves) in this tree.
We have several steps and we describe them below.
\begin{enumerate}
\item \emph{Creation of binary tree.} The binary tree $\BT$ is created as follows.
Initially the tree $\BT$ is empty.
For all $w \in W$ and all $a \in P_1(w)$ we generate the set $\Succ((w,a))$ as a
Boolean array $\Ar$ of length $|W|$ such that $\Ar[i]=1$ if $w_i \in \Succ(w,a)$ and~0
otherwise. 
We use the array $\Ar$ to add the set $\Succ((w,a))$ to $\BT$ as follows: we 
proceed from the root, if $\Ar[0]=0$ we add left edge, else the right edge, 
and proceed with $\Ar[1]$ and so on.
For every $w\in W$ and $a\in P_1(w)$, the array $\Ar$ is generated by going 
over actions in $P_2(w)$, and the addition of the set $\Succ(w,a)$ to the 
tree is achieved in $O(|W|)$ time. 
The initialization of array $\Ar$ also requires time $O(|W|)$.
Hence the total time required is $O(|W|\cdot |A_1| \cdot(|W| +|A_2|))$.
The tree has at most $|W|\cdot |A_1|$ leaves and hence the size of 
the tree is $O(|W|^2\cdot |A_1|)$.

\item \emph{The functions $f_K$, $g_K$ and $h_K$.} 
Let $\Lf$ denote the leaves of the tree $\BT$, and note that every leaf 
represents an element of $\Succ(K)$.
We do a DFT (depth-first traversal) of the tree $\BT$ and assign every leaf 
the number according to the order of leaves in which it appears in the DFT.
Hence the function $f_K$ is constructed in time $O(|W|^2\cdot |A_1|)$.
Moreover, when we construct the function $f_K$, we create an array $\GAr$ of lists 
for the function $g_K$. 
If a leaf is assigned number $i$ by $f_K$, we go from the leaf to the root
and find the set $T\in \Succ(K)$ that the leaf represents and $\GAr[i]$ is 
the list of states in $T$.
Hence the construction of $g_K$ takes time at most $O(|W|\cdot |A_1|\cdot |W|)$.
The function $h_K$ is stored as a two-dimensional array of integers with rows
indexed by numbers from $0$ to $|W|-1$, and columns by numbers $0$ to $|A_1|-1$.
For a state $w$ and action $a$, we generate the Boolean array $\Ar$, and 
use the array $\Ar$ to traverse $\BT$, obtain the leaf for $\Succ((w,a))$, 
and assign $h_K((w,a)) =f_K(\Succ(w,a))$.
It follows that $h_K$ is generated in time $O(|W|\cdot |A_1| \cdot(|W| +|A_2|))$.
  
\end{enumerate}
From the above graph construction, Proposition~\ref{prop:alt-sim},
Lemma~\ref{lemm:alt-sim}, and the linear time algorithms to solve games with 
reachability and safety objectives we have the following result for 
computing alternating simulation.

\begin{theorem}\label{thrm:alt-sim-1}
The relation $\preceq_{\altsim}$ can be computed in time 
$O(|W|\cdot |W'|\cdot |A_1|\cdot (|A_1'|\cdot |A_2'|+|A_2|) + |W|^2\cdot |A_1| 
+ |W'|^2 \cdot |A_1'|)$
for two ATSs $K$ and $K'$.
The relation $\preceq_{\altsim}$ can be computed in time 
$O(|W|\cdot |R'| + |W'|\cdot |R|)$ 
for two TSs $K$ and $K'$.
\end{theorem}

The result for the special case of TSs is obtained by noticing that 
for TSs we have both $|V|$ and $|E|$ at most $|W|\cdot |R'| +|W'| \cdot |R|$ 
(see technical details appendix for details), and our algorithm 
matches the complexity of 
the best known algorithm of~\cite{HHK95} for simulation for transition systems.
Let us denote by $n=|W|$ and $n'=|W'|$ the size of the state spaces, 
and by $m=|W| \cdot |A_1|\cdot |A_2|$ and 
$m'=|W'| \cdot |A_1'|\cdot |A_2'|$  the size of the transition 
relations. 
Then the basic algorithm requires $O(n \cdot n' \cdot m \cdot m')$ time, 
whereas our algorithm requires at most $O(m \cdot m' + n\cdot m + n'\cdot m')$ 
time, and when $n=n'$ and $m=m'$, then the basic algorithm requires
$O((n\cdot m)^2)$ time and our algorithm takes $O(m^2)$ time.

\subsection{Iterative Algorithm}\label{subsec:iterative}
In this section we will present an iterative algorithm 
for alternating simulation.
For our algorithm we will first present a new and alternative
characterization of alternating simulation through successor
set simulation.

\newcommand{\ssim}{\approxeq}

\begin{definition}[Successor set simulation]\label{defn:new_alt_sim}
Given two ATSs $K= (\Sigma,W,\whw,A_1,A_2,P_1,P_2,L,\delta)$ and 
$K'= (\Sigma,W',\whw',A_1',A_2',P_1',P_2',L',\delta')$,
a relation $\ssim \subseteq W \times W'$ is a \emph{successor set simulation} 
from $K$ to $K'$, if there exists a companion relation 
$\ssim^S \subseteq \Succ(K') \times \Succ(K)$, such that the following 
conditions hold: 
 \begin{itemize}
  \item for all $(w,w') \in \ssim$ we have $L(w) = L'(w')$; 
  \item if $(w,w') \in \ssim$, then for all actions $a \in P_1(w)$, there 
  exists an action $a' \in P_1'(w')$ such that $(\Succ(w',a'),\Succ(w,a)) 
  \in \ssim^S$; and
  \item if $(T',T) \in \ssim^S$, then for all $r' \in T'$, there exists 
  $r \in T$ such that $(r,r') \in \ssim$.
  \end{itemize}
We denote by $\ssim^*$ the maximum successor set simulation.
\end{definition}

We now show that successor set simulation and alternating simulation coincide, 
and then present the iterative algorithm to compute the maximum successor set 
simulation $\ssim^*$.

\begin{lemma}\label{lemm:alt-charac}
The following assertions hold: 
(1)~Every successor set simulation is an alternating simulation, 
and every alternating simulation is a successor set simulation.
(2)~We have $\ssim^*=\preceq_{\altsim}$.
\end{lemma}
\begin{proof}
The second assertion is an easy consequence of the first one, and 
we prove inclusion in both directions to prove the first assertion.
\begin{itemize}
\item \emph{(Alternating simulation impiles successor set simulation).} 
Suppose $\preceq$ is an alternating simulation. 
We need to prove that $\preceq$ is also a successor set simulation. 
For this we will construct the witness companion relation 
$\ssim^S \subseteq \Succ(K')\times \Succ(K)$ to satisfy 
Definition~\ref{defn:new_alt_sim}.

We define
\begin{equation*}
\ssim^S = \left \{(\Succ(w',a'),\Succ(w,a)) \mid 
\begin{matrix} (w,w') \in \preceq 
\wedge a \in P_1(w) \wedge a' \in P'_1(w'). \\ 
\forall b' \in P'_2(w') \cdot \exists b \in P_2(w) \cdot (\delta(w,a,b),\delta'(w',a',b')) \in \preceq 
\end{matrix}\right \}
\end{equation*}
      
Clearly, if $(T',T) \in \ssim^S$, then $T' = \Succ(w',a')$ and $T = \Succ(w,a)$ 
for some $(w,w') \in \preceq$ and $a \in P_1(w)$ and $a' \in P'_1(w')$ such 
that for all $b' \in P'_2(w')$ there exists $b \in P_2(w)$, such that 
$(\delta(w,a,b),\delta'(w',a',b')) \in \preceq$. 
Since every $r'$ in $T'$ is such that $r' = \delta'(w',a',b')$ for some 
$b' \in P_2'(w')$, we have that for every $r' \in T'$, there exists 
$b \in P_2(w)$, such that $(\delta(w,a,b),r') \in \preceq$. 
Hence for every $r' \in T'$, there exists $r \in T$ such that 
$(r,r') \in \preceq$.
The other requirements of Definition~\ref{defn:new_alt_sim} are trivially satisfied. 
Hence $\preceq$ is also a successor set simulation.

\item \emph{(Successor set simulation implies alternating simulation).} 
Suppose $\ssim$ is a successor set simulation. 
Hence there exists a companion relation 
$\ssim^S \subseteq \Succ(K') \times \Succ(K)$ satisfying the requirements of 
Definition~\ref{defn:new_alt_sim}. 
We need to prove that $\ssim$ is also an alternating simulation.
From Definition~\ref{defn:new_alt_sim}, for all $(w,w') \in \ssim$, 
for all $a \in P_1(w)$, there exists $a' \in P'_1(w')$ such that 
$(\Succ(w',a'),\Succ(w,a)) \in \ssim^S$. 
Now, for any $b' \in P'_2(w')$, there exists $r' \in \Succ(w',a')$, 
such that $r' = \delta'(w',a',b')$. 
Since, $(\Succ(w',a'),\Succ(w,a)) \in \ssim^S$, and 
$r' \in \Succ(w',a')$, there exists a $r \in \Succ(w,a)$ and hence there exists 
$b \in P_2(w)$ satisfying $r = \delta(w,a,b)$, such that $(r,r') \in \ssim$, 
which is same as $(\delta(w,a,b),\delta'(w',a',b')) \in \ssim$. 
Hence $\ssim$ is also an alternating simulation.
\end{itemize}
This completes the proof.
\hfill\qed
\end{proof}

We will now present our iterative algorithm to compute $\ssim^*$, and we will 
denote by $\ssim^S$ the witness companion relation of $\ssim^*$. 
Our algorithm will use the following graph construction: 
Given an ATS $K$, we will construct the graph $G_K=(V_K,E_K)$ as follows:
(1)~$V_K= W \cup \Succ(K)$, where $W$ is the set of states; and
(2)~$E_K = \set{(w,\Succ(w,a)) \mid w\in W \wedge a \in P_1(w) } \cup \set{(T,r) \mid T \in \Succ(K) \wedge r \in T }$.
The graph $G_K$ can be constructed in time $O(|W|^2\cdot |A_1|)$ using the 
binary tree data structure described earlier.
Our algorithm will use the standard notation of $\pre$ and $\post$: 
given a graph $G=(V,E)$, for a set $U$ of states, $\post(U)=\set{ v \mid \exists u \in U, 
(u,v)\in E}$ is the set of successor states of $U$, and similarly,  
$\pre(U)=\set{v \mid \exists u \in U, (v,u)\in E}$ 
is the set of predecessor states. 
If $U=\set{q}$ is singleton, we will write $\post(q)$ instead of $\post(\set{q})$.
Note that in the graph $G_K$ for the state $T \in \Succ(K)$ we have 
$\post(T)=\set{q \mid q \in T}=T$. 
Given ATSs $K$ and $K'$ our algorithm will work simultaneously on the graphs
$G_{K}$ and $G_{K'}$ using three data structures, namely, 
$\simstr$, $\countstr$ and $\remove$ for the relation $\ssim^*$ (resp. 
$\simstr^S$, $\countstr^S$ and $\remove^S$ for the companion relation 
$\ssim^S$).
The data structures are as follows: 
(1)~Intuitively $\simstr$ will be an overapproximation of $\ssim^*$, and 
will be maintained as a two-dimensional Boolean array so that whenever the $i,j$-th entry 
is false, then we have 
a witness that the $j$-th state $w_j'$ of $K'$ does not simulate the $i$-th state $w_i$ of $K$
(similary we have $\simstr^S$ over $\Succ(K')$ and $\Succ(K)$ for the relation $\ssim^S$).
(2)~The data structure $\countstr$ is two-dimensional array, such that for a 
state $w'\in W'$ and $T \in \Succ(K)$ we have $\countstr(w',T)$ is the number 
of elements in the intersection of the successor states of $w'$ and 
the set of all states that $T$ simulates according to $\simstr^S$; 
and we also have similar array $\countstr^S$ for $T,w'$ elements.
(3)~Finally, the data structure $\remove$ is a list of sets, where for every 
$w'\in W'$ we have $\remove(w')$ is a set such that every element of the set 
belongs to $\Succ(K)$. 
Similarly for every $T \in \Succ(K)$ we have $\remove^S(T)$ is a set of states.
Intuitively the interpretation of remove data structure will be as follows: 
if $T\in \Succ(K)$ belongs to $\remove(w')$, then no element $w$ of $T$ is 
simulated by $w'$. 
Our algorithm will always maintain $\simstr$ (resp. $\simstr^S$) 
as overapproximation of $\ssim^*$ (resp. $\ssim^S$), 
and will iteratively prune them.
Our algorithm is iterative and we denote by $\prevsim$ (resp. 
$\prevsim^S$) the $\simstr$ (resp. $\simstr^S$) of the previous 
iteration.
To give an intuitive idea of the invariants maintained by the 
algorithm (Algorithm~\ref{algo:iterative}) let us denote by 
$\simstr(w)$ the set of $w'$ such that $\simstr(w,w')$ is true, and 
let us denote by $\inv\simstr(w')$ the inverse of 
$\simstr(w')$, i.e., the set of states $w$ such that $(w,w')$-th element 
of $\simstr$ is true (similar notation for $\inv\prevsim(w'),
\inv\simstr^S(T)$ and $\inv\prevsim^S(T)$). 
The algorithm will ensure the following invariants at different steps:
\begin{enumerate}
\item For $w \in W, w'\in W'$ and $T\in \Succ(K), T' \in \Succ(K')$,
  \begin{enumerate}
  \item if $\simstr(w,w')$ is false, then $(w,w')\notin \ssim^*$; 
  \item similarly, if $\simstr^S(T',T)$ is false, then $(T',T) \notin \ssim^S$.
  \end{enumerate}
\item For $w' \in W'$ and $T \in \Succ(K)$,
  \begin{enumerate}
  \item $\countstr(w',T) = |\post(w') \cap \inv\simstr^S(T)|$; and
  \item $\countstr(T,w') = |\post(T) \cap \inv\simstr(w')|=|T \cap \inv\simstr(w')|$
  \end{enumerate}
\item For $w' \in W'$ and $T \in \Succ(K)$,
  \begin{enumerate}
  \item $\remove(w') = \pre(\inv\prevsim(w')) \setminus \pre(\inv\simstr(w'))$
  \item $\remove(T) = \pre(\inv\prevsim^{S}(T)) \setminus \pre(\inv\simstr^{S}(T))$.
  \end{enumerate}
\end{enumerate}

The algorithm has two phases: the initialization phase, where the data structures are 
initialized; and then a while loop. 
The while loop consists of two parts: one is pruning of $\simstr$ and the 
other is the pruning of $\simstr^S$ and both the pruning steps are similar.
The initialization phase initializes the data structures and is described in 
Steps~1, 2, and~3 of Algorithm~\ref{algo:iterative}. 
Then the algorithm calls the two pruning steps in a while loop. 
The condition of the while loop checks whether $\prevsim$ and 
$\simstr$ are the same, and it is done in constant time 
by simply checking whether $\remove$ is empty. 
We now describe one of the pruning procedures and the other is similar.
The pruning step is similar to the pruning step of the algorithm 
of~\cite{HHK95} for simulation on transition systems.
We describe the pruning procedure {\sc PruneSimStrSucc}.
For every state $w' \in W'$ such that the set $\remove(w')$ is non-empty,
we run a for loop.
In the for loop we first obtain the predecessors $T'$ of $w'$ in $G_{K'}$ 
(each predecessor belongs to $\Succ(K')$) and an element $T$ from 
$\remove(w')$. 
If $\simstr^S(T',T)$ is true, then we do the following steps:
(i)~We set $\simstr^S(T',T)$ to false, because we know that there does not 
exist any element $w \in T$ such that $w'$ simulates $w$.
(ii)~Then for all $s'$ that are predecessors of $T'$ in $G_{K'}$ we decrement 
$\countstr(s',T)$, and if the count is zero, then we add $s'$ to the remove 
set of $T$.
Finally we set the remove set of $w'$ to $\emptyset$.
The description of {\sc PruneSimStr} to prune $\simstr$ is similar.


\begin{algorithm*}[t]
\caption{Iterative Algorithm}
\label{algo:iterative}
{
\begin{tabbing}
aaa \= aaa \= aaa \= aaa \= aaa \= aaa \= aaa \= aaa \kill
\> {\bf Input:} $K= (\Sigma,W,\whw,A_1,A_2,P_1,P_2,L,\delta)$, $K'= (\Sigma,W',\whw',A_1',A_2',P_1',P_2',L',\delta')$. \\
\>   {\bf Output:} $\ssim^*$. \\ 

\> 1. {\bf Initialize $\simstr$ and $\simstr^S$:} \\ 
\>\> 1.1. {\bf for all} $w \in W, w'\in W'$ \\
\>\>\>\> $\prevsim(w,w')\gets$ true; \\
\>\>\>\> {\bf if } $L(w)=L'(w')$, then $\simstr(w,w')\gets$ true; \\
\>\>\>\> {\bf else } $\simstr(w,w')\gets$ false; \\

\>\> 1.2. {\bf for all} $T \in \Succ(K)$ and $T'\in \Succ(K')$ \\ 
\>\>\>\> $\prevsim^S(T',T)=\simstr^S(T',T)\gets $ true; \\

\> 2. {\bf Initialize $\countstr$ and $\countstr^S$:} \\ 
\>\> 2.1. {\bf for all} $w'\in W'$ and $T \in \Succ(K)$ \\
\>\>\>\> $\countstr(w',T) \gets |\post(w') \cap \inv\simstr^{S}(T)| = |\post(w')|$; \\
\>\>\>\> $\countstr^S(T,w') \gets |\post(T) \cap \inv\simstr(w')|$; \\

\> 3. {\bf Initialize $\remove$ and $\remove^S$:} \\ 
\>\> 3.1. {\bf for all} $w'\in W'$ \\
\>\>\>\> $\remove(w')\gets \Succ(K)\setminus \pre(\inv\simstr(w'))$; \\
\>\> 3.2. {\bf for all} $T \in \Succ(K)$ \\
\>\>\>\> $\remove^S(T)\gets \emptyset$; \\

\>  {\bf Pruning while loop:} \\
\> 4. {\bf while} $\prevsim \neq \simstr$ \\
\>\> 4.1 $\prevsim \gets \simstr$; \\
\>\> 4.2 $\prevsim^S \gets \simstr^S$; \\
\>\> 4.3 {\bf Procedure} {\sc PruneSimStrSucc}; \\
\>\> 4.4 {\bf Procedure} {\sc PruneSimStr}; \\

\> 5. {\bf return} $\set{(w,w') \in W\times W'\mid \simstr(w,w') \text{ is true}}$;

\end{tabbing}
}
\end{algorithm*}

\begin{algorithm*}[t]
\caption{Procedure PruneSimStrSucc}
\label{algo:prunesucc}
{
\begin{tabbing}
aaa \= aaa \= aaa \= aaa \= aaa \= aaa \= aaa \= aaa \= aa \= aaa \kill

\> 1. {\bf forall } $w' \in W'$ such that $\remove(w') \neq \emptyset$  \\
\>\> 1.1. {\bf forall} $T' \in \pre(w')$ and $T \in \remove(w')$ \\
\>\>\> 1.1.1 {\bf if} ($\simstr^S(T',T)$) \\
\>\>\>\>\>  $\simstr^S(T',T) \gets$ false; \\
\>\>\>\>\> 1.1.1.A. {\bf forall }  ($s' \in \pre(T')$) \\
\>\>\>\>\>\>\>\>  $\countstr(s',T) \gets \countstr(s',T) - 1$; \\
\>\>\>\>\>\>\>\>   {\bf if}  ($\countstr(s',T) = 0$ ) \\
\>\>\>\>\>\>\>\>\>  $\remove^S(T) \gets \remove^S(T) \cup \set{s'}$; \\
\>\> 1.2. $\remove(w') \gets \emptyset$; 
\end{tabbing}
}
\end{algorithm*}

\begin{algorithm*}[t]
\caption{Procedure PruneSimStr}
\label{algo:prune}
{
\begin{tabbing}
aaa \= aaa \= aaa \= aaa \= aaa \= aaa \= aaa \= aaa \= aa \= aaa \kill

\> 1. {\bf forall } $T \in \Succ(K)$ such that $\remove^S(T) \neq \emptyset$  \\
\>\> 1.1. {\bf forall} $w \in \pre(T)$ and $w' \in \remove^S(T)$ \\
\>\>\> 1.1.1 {\bf if} ($\simstr(w,w')$) \\
\>\>\>\>\>  $\simstr(w,w') \gets$ false; \\
\>\>\>\>\> 1.1.1.A. {\bf forall }  ($D \in \pre(w)$) \\
\>\>\>\>\>\>\>\>  $\countstr^S(D,w') \gets \countstr^S(D,w') - 1$; \\
\>\>\>\>\>\>\>\>   {\bf if}  ($\countstr^S(D,w') = 0$ ) \\
\>\>\>\>\>\>\>\>\>  $\remove(w') \gets \remove(w') \cup \set{D}$; \\
\>\> 1.2. $\remove^S(T) \gets \emptyset$; 
\end{tabbing}
}
\end{algorithm*}

\smallskip\noindent{\bf Correctness.} 
Our correctness proof will be in two steps. 
First we will show that invariant~1  (both about $\simstr$ and $\simstr^S$) and 
invariant~2 (both about $\countstr$ and $\countstr^S$) are true
at the beginning of step 4.1.
The invariant~3.(a) (on $\remove$) is true after the procedure call 
{\sc PruneSimStr} (step~4.4) and invariant 3.(b) (on $\remove^S$) is true after 
the procedure call {\sc PruneSimStrSucc} (step~4.3). 
We will then argue that these invariants ensure correctness of the algorithm.

\smallskip\noindent{\em Maintaining invariants.}
We first consider invaraint~1, and focus on invariant~1.(b) (as
the other case is symmetric).
In procedure {\sc PruneSimStrSucc} when we set $\simstr^S(T',T)$ to
false, we need to show that $(T',T) \not\in \ssim^S$.
The argument is as follows: when we set $\simstr^S(T',T)$ to
false, we know that since $T \in \remove(w')$ we have 
$\countstr^S(T,w')=0$ (i.e., $\post(T) \cap \inv\simstr(w')=\emptyset$).
This implies that for every $w \in T$ we have that $w'$ does not simulate 
$w$. Also note that since $\countstr^S$ is never incremented, if it 
reaches zero, it remains zero. 
This proves the correctness of invariant~1.(b) (and similar argument holds 
for invariant~1.(a)).
The correctness for invariant 2.(a) and 2.(b) is as follows: whenever we decrement
$\countstr(s',T)$ we have set $\simstr^S(T',T)$ to false, and  
$T'$ was earlier both in $\post(s')$ as well as in $\inv\simstr^S(T)$, 
and is now removed from $\inv\simstr^S(T)$. 
Hence from the set $\post(s') \cap \inv\simstr^S(T)$ we remove the element $T'$ 
and its cardinality decreases by~1. 
This establishes correctness of invariant~2.(a) (and invariant~2.(b) is similar).
Finally we consider invariant 3.(a): when we add 
$s'$ to $\remove^S(T)$, then we know that $\countstr(s',T)$ was decremented
to zero, which means $T'$ belongs to $\inv\prevsim^S(T)$, but not to 
$\inv\simstr^S(T)$.
Thus $s'$ belongs to $\pre(\inv\prevsim^S(T))$ (since $s'$ belongs to $\pre(T')$),
but not to $\pre(\inv\simstr^S(T))$.
This shows that $s'$ belongs to $\remove^S(T)$, and establishes correctness
of the desired invariant (argument for invariant 3.(b) is similar). 

\smallskip\noindent{\em Invariants to correctness.}
The initialization part ensures that $\simstr$ is an overapproximation of 
$\ssim^*$ and it follows from invariant~1 that the output is an overapproximation of 
$\ssim^*$.
Similarly we also have that $\simstr^S$ in the end is an overapproximation of 
$\ssim^S$. 
To complete the correctness proof, let $\simstr$ and $\simstr^S$ be the 
result when the while loop iterations end. 
We will now show that $\simstr$ and $\simstr^S$ are witness to satisfy successor set 
simulation.
We know that when the algorithm terminates, 
$\remove(w') = \emptyset$ for every $w' \in W'$, and 
$\remove^S(T) = \emptyset$ for every $T \in \Succ(K)$ (this follows 
since $\simstr=\prevsim$). 
To show that $\simstr$ and $\simstr^S$ are witnesses to satisfy successor set simulation, 
we need to show the following two properties:
(i)~If $\simstr(w,w')$ is true, then for every $a \in P_1(w)$, there exists $a' \in P_1'(w')$ 
such that $\simstr^S(\Succ(w',a'),\Succ(w,a))$ is true.
(ii)~If $\simstr^S(T',T)$ is true, then for every $s' \in T'$, there exists 
$s \in T$ such that $\simstr(s,s')$ is true.
The property (i)~holds because for every $a \in P_1(w)$, we have that 
$\countstr(w',T) > 0$, where $T= \Succ(w,a)$, 
(because otherwise, $w'$ would have been inserted in $\remove(T)$, but since $\remove(T)$ is empty, 
consequently $\simstr(w,w')$ must have been made false). 
Hence we have that $\post(w') \cap \inv\simstr^S(T)$ 
is non-empty and hence there exists $T' \in \post(w')$ such that 
$\simstr^S(T',T)$ is true. Similar argument works for (ii).
Thus we have established that $\simstr$ is both an overapproximation 
of $\ssim^*$ and also a witness successor set relation.
Since $\ssim^*$ is the maximum successor set relation, it follows
that Algorithm~\ref{algo:iterative} correctly computes 
$\ssim^*=\preceq_{\altsim}$ ($\ssim^*=\preceq_{\altsim}$ 
by Lemma~\ref{lemm:alt-charac}).

\smallskip\noindent{\em Space complexity.} We now argue that the space complexity 
of the iterative algorithm is superior as compared to the game based algorithm.
We first show that the space taken by Algorithm~\ref{algo:iterative} 
is $O(|W|^2\cdot |A_1| + |W'|^2\cdot |A_1'| + 
|W|\cdot |W'|\cdot |A_1|\cdot |A_1'|)$.
For the iterative algorithm, the space requirements are,
(i)~$\simstr$ and $\simstr^S$ require at most $O(|W|\cdot|W'|)$ and 
$O(|W|\cdot |W'|\cdot |A_1|\cdot |A_1'|)$ space, respectively;
(ii)~$\countstr$ and $\countstr^S$ require at most $O(|W|\cdot |W'|\cdot |A_1|)$ 
space each;
(iii)~$\remove$ and $\remove^S$ maintained as an array of sets require 
at most $O(|W|\cdot |W'|\cdot |A_1|)$, space each.
Also, for the construction of graphs $G_K$ and $G_{K'}$ using the binary tree data structure 
as described earlier, the space required is at most $O(|W|^2\cdot|A_1|)$ and 
$O(|W'|^2\cdot|A_1'|)$, 
respectively.
As compared to the space requirement of the iterative algorithm, the game 
based algorithm requires to store the entire game graph and requires at least 
$O(|W|\cdot |W'|\cdot |A_1|\cdot |A_1'|\cdot |A_2'|)$ space (to store edges 
in $E_3$) as well as space $O(|W|^2\cdot |A_1| + |W'|^2 \cdot |A_1'|)$ for the 
binary tree data structure.
The iterative algorithm can be viewed as an efficient simultaneous pruning 
algorithm that does not explicitly construct the game graph (and thus 
saves at least a factor of $|A_2'|$ in terms of space).
We now show that the iterative algorithm along with being space efficient 
matches the time complexity of the game based algorithm.



\smallskip\noindent{\em Time complexity.}
The data structures for $\simstr$ (also $\simstr^S$) and $\countstr$ (also 
$\countstr^S$) are as described earlier. 
We store $\remove$ and $\remove^S$ as a list of lists (i.e., it is a list of 
sets, and sets are stored as lists). 
It is easy to verify that all the non-loop operations take unit cost, 
and thus for the time complexity, we need to estimate the number of times 
the different loops could run. 
Also the analysis of the initialization steps are straight forward, and
we present the analysis of the loops below:
(1)~The {\bf while} loop (Step~4) of Algorithm \ref{algo:iterative} can run for 
at most $|W|\cdot|W'|$ 
iterations because in every iteration (except the last iteration) at least 
one entry of $\simstr$ changes from true to false (otherwise the iteration stops), 
and $\simstr$ has $|W|\cdot|W'|$-entries.
(2)~The {\bf forall} loop (Step~1) in Algorithm \ref{algo:prunesucc} can overall 
run for at most $|W'|\cdot |W| \cdot |A_1|$ iterations. 
This is because elements of $\remove(w')$ are from $\Succ(K)$ and elements $T$ from 
$\Succ(K)$ are included in $\remove(w')$ at most once (when $\countstr^S(T,w')$ is set 
to zero, and once $\countstr^S(T,w')$ is set to zero, it remains zero).
Thus $\remove(w')$ can be non-empty at most $|\Succ(K)|$ times, and 
hence the loop runs at most $|W| \cdot |A_1|$ times for states $w' \in W'$. 
(3)~The {\bf forall} loop (Step~1.1) in Algorithm~\ref{algo:prunesucc}   
can overall run for at most $|W'|\cdot |A_1'|\cdot |A_2'|\cdot |W|\cdot |A_1|$ 
iterations.
The reasoning is as follows: for every edge $(T',w') \in G_{K'}$ and
$T\in \Succ(K)$ the loop runs at most once (since every $T$ is included in 
$\remove(w')$ at most once).
Hence the number of times the loop runs is at most the number of edges 
in $G_{K'}$ (at most $|W'|\cdot |A_1'|\cdot |A_2'|$) times the number of elements in 
$\Succ(K)$ (at most $|W|\cdot |A_1|$).
Thus overall the number of iterations of Step~1.1 of Algorithm~\ref{algo:prunesucc}
is at most $|W'|\cdot |A_1'|\cdot |W|\cdot |A_1|$.
(4)~The {\bf forall} loop (Step 1.1.1.A) in Algorithm~\ref{algo:prunesucc} 
can overall run for at most $|W'|\cdot |A_1'| \cdot |A_2'|\cdot |W|\cdot|A_1|$ iterations 
because every edge $(s',T')$ in $G_{K'}$ would be iterated at most once for every 
$T \in \Succ(K)$ (as for every $T,T'$ we set $\simstr^S(T,T')$ false at most
once, and the loop gets executed when such an entry is set to false).
The analysis of the following items (5), (6), and (7), are similar to 
(2), (3), and (4), respectively.
(5)~The {\bf forall} loop (Step~1) in Algorithm~\ref{algo:prune} can overall run for 
at most $|W|\cdot |A_1|\cdot |W'|$ iterations, because $\remove^S(T)$ can be 
non-empty at most $|W'|$ times (i.e., the number of different $T$ is 
at most $|\Succ(K)|=|W|\cdot |A_1|$).
(6)~The {\bf forall} loop (Step~1.1) in Algorithm~\ref{algo:prune}  
can overall run for at most $|W|\cdot |A_1| \cdot |A_2|\cdot |W'|$ 
iterations because every edge $(w,T)$ in $G_K$ can be iterated over at most 
once for every $w'$ (the number of edges in $G_K$ is $|W|\cdot |A_1| \cdot |A_2|$
and number of $w'$ is at most $|W'|$). 
(7)~The {\bf forall} loop (Step~1.1.1.A) in Algorithm~\ref{algo:prune} can overall run for 
at most $|W|\cdot |A_1|\cdot |A_2|\cdot |W'|$ iterations because every edge 
$(w,D)$ in $G_K$ would be iterated over at most once for every $w' \in W'$.
Adding the above terms, we get that the total time complexity is 
$O\big(|W|\cdot |W'|\cdot |A_1|\cdot (|A_1'|\cdot |A_2'| + |A_2|)\big)$, 
i.e., the time complexity matches the time complexity of the game reduction 
based algorithm.
We also tabulate our analysis in Table~\ref{tab:complexities}.
We also remark that for transition systems (TSs), the procedure 
{\sc PruneSimStrSucc} coincides with {\sc PruneSimStr} and our algorithm 
simplifies to the algorithm of~\cite{HHK95}, and thus matches the 
complexity of computing simulation for TSs.

\begin{theorem}
Algorithm~\ref{algo:iterative} correctly computes $\preceq_{\altsim}$ 
in time $O\big(|W|\cdot |W'|\cdot |A_1|\cdot (|A_1'|\cdot |A_2'| + |A_2|) 
+ |W|^2\cdot |A_1| + |W'|^2\cdot|A_1'|\big)$. 
\end{theorem}


\begin{table}[t]
\begin{center}
\begin{footnotesize}
\begin{tabular}{|c|c|c|}
\hline
Step & Complexity & Justification\\
\hline
\hline
\textbf{while} loop (Step~4 of Algorithm~\ref{algo:iterative})  
& $O(|W|\cdot|W'|)$ & 
\ all (except the last) iteration changes at least one of \  \\
& &  the $|W|\cdot |W'|$-entries of $\simstr$ from true to false \\
\hline
\hline
\textbf{forall} loop (Step~1 of Algorithm~\ref{algo:prunesucc}) 
& $O(|W'|\cdot |W|\cdot |A_1|)$ & 
$\remove(w')$ can be non-empty only\\
& & $|\Succ(K)|$ times, for each $w' \in W'$\\
\hline
\textbf{forall} loop  (Step~1.1 of Algorithm~\ref{algo:prunesucc})
&\ $O(|\Succ(K)|\cdot |W'|\cdot |A_1'| \cdot |A_2'|)$ \   
& every edge in $G_{K'}$ can be iterated at most once \\
& & for each $T$ in $\Succ(K)$, and number of edges  \\
& & in $G_{K'}$ is 
$|W'|\cdot |A_1'| \cdot |A_2'|$ \\
\hline
\textbf{forall} loop (Step~1.1.1.A of Algorithm~\ref{algo:prunesucc})
& $O(|\Succ(K)|\cdot |W'|\cdot |A_1'| \cdot |A_2'|)$ 
& every edge in $G_{K'}$ can be iterated at most once \\
& & for each $T$ in $\Succ(K)$, and number of edges  \\
& & in $G_{K'}$ is 
$|W'|\cdot |A_1'| \cdot |A_2'|$ \\
\hline
\hline
\textbf{forall} loop (Step~1 of Algorithm~\ref{algo:prune}) 
& $O(|W'|\cdot |W|\cdot |A_1|)$ & 
$\remove^S(T)$ can be non-empty only\\
& & $|W'|$ times, for each $T \in \Succ(K)$\\
\hline
\textbf{forall} loop  (Step~1.1 of Algorithm~\ref{algo:prune})  
& $O(|W'|\cdot |W|\cdot |A_1| \cdot |A_2|)$ 
& every edge in $G_{K}$ can be iterated at most once \\
& & for each $w'$ in $W'$, and number of edges  \\
& & in $G_{K}$ is $|W|\cdot |A_1| \cdot |A_2|$ \\
\hline
\ \textbf{forall} loop (Step~1.1.1.A of Algorithm~\ref{algo:prune}) \ 
& $O(|W'|\cdot |W|\cdot |A_1| \cdot |A_2|)$ 
& every edge in $G_{K}$ can be iterated at most once \\
& & for each $w'$ in $W'$, and number of edges  \\
& & in $G_{K}$ is $|W|\cdot |A_1| \cdot |A_2|$ \\
\hline
\hline
\end{tabular}
\end{footnotesize}
\caption{Loop-wise complexity}\label{tab:complexities}
\end{center}
\end{table}

\section{Conclusion}
 In this work we presented faster algorithms for alternating simulation 
and alternating fair simulation which are core algorithmic problems in 
analysis of composite and open reactive systems, as well as state space 
reduction for graph games (that has deep connection with automata theory and
logic).
Moreover, our algorithms are obtained as efficient reductions to graph games 
with reachability and parity objectives with three priorities, and 
efficient implementations exist for all these problems (for example, 
see~\cite{Lange09} for implementation of games with reachability and 
parity objectives, and~\cite{dAF07} for specialized implementation of games 
with parity objectives with three priorities).

\clearpage
\section*{Technical Details Appendix}

\section{Fair Alternating Simulation}

We now present the reduction to parity games with three priorities for the 
special case of fair simulation.
Given the fair TSs $\mathcal{K}=(K,F)$ and $\mathcal{K}'=(K',F')$, 
we construct the game graph $G=((V,E),(V_1,V_2))$ is as follows:
\begin{itemize} \renewcommand{\labelitemi}{$-$}
  \item \emph{Player~1 vertices:} $V_1 = \set{\structure{w,w'} \mid w \in W, w' \in  W', L(w) = L'(w') } \cup \set{\frownie}$.
  \item \emph{Player 2 vertices:} $V_2 = (W \times W' \times \set{\$}) $ 
  \item \emph{Edges:} The edge set $E$ is as follows:
    \begin{eqnarray*}
      E &=& \set{(\structure{w_1,w_2},\structure{w'_1,w_2,\$}) \mid \structure{w_1,w_2}\in V_1, (w_1,w_1') \in R }\\
      & & \cup \set{(\structure{w_1',w_2,\$},\structure{w_1',w_2'}) \mid \structure{w_1',w_2,\$} \in V_2, (w_2,w_2') \in R', \structure{w_1',w_2'} \in V_1 }\\
      & & \cup \set{(\structure{w_1',w_2,\$},\frownie)              \mid \structure{w_1',w_2,\$} \in V_2, \forall w_2' \text{ if } (w_2,w_2')\in R', \text{ then } \structure{w_1',w_2'} \not\in V_1  }\\
      & & \cup \set{(\frownie,\frownie)}
    \end{eqnarray*}
  \end{itemize}
The three-priority parity objective $\Phi^*$ for player~2 with the priority function 
$p$ is specified as follows: 
for vertices $v \in (W \times F')\cap V_1$ we have $p(v)=0$; 
for vertices $v \in ((F \times W' \setminus W \times F')\cap V_1) \cup \set{\frownie}$ we have $p(v)=1$;
and all other vertices have priority~2.
Also without loss of generality we assume that for every $w\in W$ there exists
a fair run from $w$.
The specialization of Proposition~\ref{prop:alt_fair_sim_correctness} gives 
us the following proposition.

\begin{proposition}~\label{prop:fair_sim_correctness}
Let $\Win_2= \set{(w_1,w_2) \mid  \structure{w_1,w_2} \in V_1, \structure{w_1,w_2} \in W_2(\Phi^*), \text{i.e., there is a winning state for player 2}}$.
Then we have $\Win_2= \preceq_{\fair}$.
\end{proposition}

\begin{lemma}\label{lem_fair_size}
For the game graph constructed for fair simulation we have 
$|V_1| + |V_2| \leq  O(|W|\cdot |W'|)$; and 
$|E| \leq O(|W|\cdot |R'| +  |W'|\cdot |R|)$.
\end{lemma}
\begin{proof}
We have  $|V_1| \le |V_2| +1 = |W|\cdot|W'| + 1= O(|W|\cdot |W'|)$.
We have
\[
|E| \le 1 + 2\cdot |W|\cdot |W'| + \sum_{w' \in W'} \sum_{\begin{subarray}{c} w \in W \\ L(w) = L'(w') \end{subarray}} \deg(w) + \sum_{w \in W} \sum_{w' \in W'} \deg(w')
  \le 1 + 2\cdot |W|\cdot |W'| + |W'|\cdot|R| + |W|\cdot|R'|,
\]
where $\deg(w)$ (resp. $\deg(w')$) denotes the number of outedges (or out-degree) of $w$ (resp. $w'$). 
The result follows.
\hfill\qed
\end{proof}

The reduction and the results to solve parity games with three priorities establish that 
$\preceq_{\fair}$ can be computed in time 
$O(|W|\cdot |W'|\cdot (|W|\cdot |R'| +  |W'|\cdot |R|))$.
This completes the last item of Theorem~\ref{thrm:alt_fair}.

\section{Alternating Simulation}

\subsection{Improved Algorithm Through Games}
In this section we consider the specialization of the 
alternating simulation algorithm for TS. 
Since we have already established in Section~\ref{sec:alt-sim-1} that 
the game graph construction complexity is linear in the size 
of the game graph, we only need to estimate the size of the vertex set
and the edge set for TS.

\begin{lemma}\label{lemm:alt-sim-ts}
For the game graph constructed for alternating simulation for TS, we have 
$|V_1| + |V_2| \leq O(|W|\cdot |W'|\cdot |A_1|\cdot |A_1'|)$ and 
$|E| \leq O(|W|\cdot |W'|\cdot (|A_1| + |A_1'|))$.
\end{lemma} 
\begin{proof}
Note that the size of the vertex set is bounded by the same quantity as for 
the general case for ATS, and thus the vertex size bound is trivial.
We now consider the case for edges.
First observe that since $|A_2|=1$, it follows that 
$\Succ(K) \leq |W|$ as every $\Succ((w,a))$ is singleton (i.e., a state), and 
hence $\Succ(K)$ has at most $|W|$ elements and each element is a singleton 
state.
Similarly we have $\Succ(K')\leq |W'|$. 
We have $|E| = |E_1| + |E_2| + |E_3| + |E_4|$, and we obtain bounds for
them below:
\[
|E_1| = \sum_{w' \in W'} \sum_{w \in W} |P_1(w)| \le |W'|\cdot |W|\cdot |A_1|
\]
\[
|E_2| = \sum_{T \in \Succ(K)} \sum_{w' \in W'} |P'_1(w')| \le 
|\Succ(K)|\cdot |W'|\cdot |A_1'|
\le |W|\cdot |W'|\cdot |A_1'|
\]
\[
|E_3| = \sum_{T \in \Succ(K)} \sum_{T' \in \Succ(K')} |T'|
\le |\Succ(K)|\cdot |\Succ(K')|\le |W|\cdot |W'|
\]
\[
|E_4| = \sum_{r' \in W'} \sum_{T \in \Succ(K)} |T|
\le |W'|\cdot |\Succ(K)| \le |W'|\cdot |W|
\]
where in the bound for $E_3$ we used $|T'| \le |A_2|=1$ and in the bound for 
$E_4$ we used $|T| \le |A_2|=1$.
It follows that 
$|E| \leq O(|W|\cdot |W'|\cdot (|A_1| + |A_1'|))$.
and the desired result follows.
\hfill\qed
\end{proof}

Since $|R|=|W|\cdot |A_1|$ and $|R'|=|W'|\cdot |A_1'|$, we obtain the 
last result of Theorem~\ref{thrm:alt-sim-1}.

\end{document}